\documentclass[11pt,a4paper]{article}

\usepackage{amsmath,amssymb,amsthm}
\usepackage{amsfonts}
\usepackage{mathtools}
\usepackage{booktabs}
\usepackage{multirow}
\usepackage{enumitem}
\usepackage{array}
\usepackage[margin=2.5cm]{geometry}
\usepackage[colorlinks=true,linkcolor=blue,citecolor=red,urlcolor=blue]{hyperref}
\allowdisplaybreaks

\newtheorem{theorem}{Theorem}[section]
\newtheorem{corollary}[theorem]{Corollary}
\newtheorem{lemma}[theorem]{Lemma}
\newtheorem{proposition}[theorem]{Proposition}
\theoremstyle{definition}
\newtheorem{definition}[theorem]{Definition}
\newtheorem{remark}[theorem]{Remark}
\newtheorem{example}[theorem]{Example}

\newcommand{\Z}{\mathbb{Z}}
\newcommand{\C}{{\mathcal C}}
\newcommand{\cH}{{\mathcal H}}
\newcommand{\cA}{{\mathcal A}}
\newcommand{\cB}{{\mathcal B}}
\newcommand{\cM}{{\mathcal M}}
\newcommand{\cR}{{\mathcal R}}

\newcommand{\cV}{{\mathcal V}}
\newcommand{\zero}{{\mathbf{0}}}
\newcommand{\one}{{\mathbf{1}}}
\newcommand{\two}{{\mathbf{2}}}
\newcommand{\four}{{\mathbf{4}}}
\newcommand{\six}{{\mathbf{6}}}
\newcommand{\uu}{\mathbf{u}}
\newcommand{\vv}{\mathbf{v}}
\newcommand{\ww}{\mathbf{w}}
\newcommand{\xx}{\mathbf{x}}
\newcommand{\yy}{\mathbf{y}}
\newcommand{\zz}{\mathbf{z}}
\newcommand{\wt}{{\rm wt}}
\newcommand{\rank}{\operatorname{rank}}
\newcommand{\kernel}{\operatorname{ker}}
\newcommand{\ord}{\operatorname{o}}
\newcommand{\K}{\operatorname{K}}
\newcommand{\Span}{\operatorname{span}}
\newcommand{\sym}{\operatorname{sym}}
\newcommand{\lab}{\operatorname{lab}}

\begin{document}

\title{\bf The kernel-block rank profile and a complete classification of
$\mathbb{Z}_2\mathbb{Z}_4\mathbb{Z}_8$-linear Hadamard codes\thanks{This work has been granted by the Juan de la Cierva 2024 grant (JDC2024-053082-I), funded by Ministerio de Ciencia, Innovación y Universidades (MICIU/AEI/10.13039/501100011033) and co-funded by the European Social Fund Plus (FSE+).}}

\author{
Dipak K. Bhunia\\
\small Departament de Matem\`atiques,\\ \small Universitat Polit\`ecnica de Catalunya,\\
\small Barcelona, Spain
}
\date{}

\maketitle

\begin{abstract}
The $\mathbb{Z}_2\mathbb{Z}_4\mathbb{Z}_8$-additive codes are subgroups of
$\mathbb{Z}_2^{\alpha_1}\times\mathbb{Z}_4^{\alpha_2}\times\mathbb{Z}_8^{\alpha_3}$, and a
$\mathbb{Z}_2\mathbb{Z}_4\mathbb{Z}_8$-linear Hadamard code is a Hadamard code which is the Gray map image of such a code.  A recursive construction of $\mathbb{Z}_2\mathbb{Z}_4\mathbb{Z}_8$-additive Hadamard codes $\mathcal H^{t_1,t_2,t_3}$ of type $(\alpha_1,\alpha_2,\alpha_3;t_1,t_2,t_3)$, with $\alpha_1\neq0$, $\alpha_2\neq0$, $\alpha_3\neq0$, $t_1\geq1$, $t_2\geq0$, and $t_3\geq1$, is known, as are the linearity, the dimension of the kernel, and the rank of the corresponding
$\mathbb{Z}_2\mathbb{Z}_4\mathbb{Z}_8$-linear Hadamard codes $H^{t_1,t_2,t_3}$ of length $2^t$, where $t+1=3t_1+2t_2+t_3$. Yet these invariants do not completely classify the family. Two infinite families of pairs of codes of distinct types share the length, the rank and
the dimension of the kernel, and for the previously classified lengths
$2^t$, with $3\leq t\leq11$, members of such pairs were separated only by computer equivalence tests.

In this paper, we introduce an equivalence invariant that resolves these
cases.  The kernel of a binary code containing the zero word partitions the
set of binary coordinates into blocks, two coordinates lying in the same block
when every kernel word takes the same value in both of them; the
\emph{kernel-block rank profile} is then the multiset of the dimensions of the
linear span punctured on those blocks.  Unlike the rank and the dimension of
the kernel, which are global, this invariant records how much of the span
survives on each of the canonical blocks into which the kernel divides the
coordinate set.  We compute it for the whole family: the kernel partition has
$2^{t_1+t_2+t_3-1}$ blocks, all of size $2^{2t_1+t_2}$, and the profile takes
at most the two values $t_2+\binom{t_1+2}{2}$ and $t_2+2+\binom{t_1+1}{2}$,
whose difference is exactly $t_1-1$; it is constant precisely when $t_1=1$.
Hence, the profile recovers $t_1$, and the length together with the dimension of the
kernel then recovers $t_2$ and $t_3$.  Two codes of the family with the same
length are therefore equivalent if and only if their types coincide, and the
number of pairwise nonequivalent such codes of length $2^t$ is
$\lfloor(t^2+6)/12\rfloor$ for every $t\geq3$. 
\end{abstract}

\medskip
\noindent\textbf{Keywords:} Hadamard code, Gray map, $\Z_2\Z_4\Z_8$-linear code, $\Z_2\Z_4\Z_8$-additive code, rank, kernel, kernel-block rank profile,
classification.

\smallskip
\noindent\textbf{2020 Mathematics Subject Classification:} 94B25, 94B60.

%%%%%%%%%%%%%%%%%%%%%%%%%%%%%%%%%%%%%%%%%%%%%%%%%%%%%%%%%%%%%%%%%%%%%%%%%%%%%
\section{Introduction}\label{sec:intro}
%%%%%%%%%%%%%%%%%%%%%%%%%%%%%%%%%%%%%%%%%%%%%%%%%%%%%%%%%%%%%%%%%%%%%%%%%%%%%

Let $\Z_{2^s}$ be the ring of integers modulo $2^s$ with $s\geq1$.  The set of
$n$-tuples over $\Z_{2^s}$ is denoted by $\Z_{2^s}^n$, and in this paper the
elements of $\Z_{2^s}^n$ are also called vectors.  A code over $\Z_2$ of
length $N$ is a nonempty subset of $\Z_2^N$, and it is linear if it is a
subspace of $\Z_2^N$.  Similarly, a nonempty subset of $\Z_{2^s}^n$ is a
$\Z_{2^s}$-additive code if it is a subgroup of the additive group of
$\Z_{2^s}^n$.  A $\Z_2\Z_4\Z_8$-additive code is a subgroup of
$\Z_2^{\alpha_1}\times\Z_4^{\alpha_2}\times\Z_8^{\alpha_3}$.  Note that a
$\Z_2\Z_4\Z_8$-additive code is a linear code over $\Z_2$ when
$\alpha_2=\alpha_3=0$, a $\Z_4$-additive or $\Z_8$-additive code when
$\alpha_1=\alpha_3=0$ or $\alpha_1=\alpha_2=0$, respectively, and a
$\Z_2\Z_4$-additive code when $\alpha_3=0$.  The order of a vector
$\uu\in\Z_{2^s}^n$, denoted by $\ord(\uu)$, is the smallest positive integer
$m$ such that $m\uu=(0,\dots,0)$.  Also, the order of a vector
$\uu\in\Z_2^{\alpha_1}\times\Z_4^{\alpha_2}\times\Z_8^{\alpha_3}$, denoted by
$\ord(\uu)$, is the smallest positive integer $m$ such that
$m\uu=(0,\dots,0\mid 0,\dots,0\mid 0,\dots,0)$.  Throughout the paper, vectors
are written in boldface, and $\zero$ and $\one$ denote the all-zero and the
all-one vector, the length being always clear from the context.

Two binary codes $C_1$ and $C_2$ of length $N$ are said to be equivalent if
there are a vector $\uu\in\Z_2^N$ and a permutation of coordinates $\pi$ such
that $C_2=\{\uu+\pi(\xx):\xx\in C_1\}$.  The Hamming weight of a vector
$\uu\in\Z_2^N$, denoted by $\wt_H(\uu)$, is the number of its nonzero
coordinates, and the Hamming distance $d_H(\uu,\vv)$ of two vectors is the
number of coordinates in which they differ; hence,
$d_H(\uu,\vv)=\wt_H(\uu-\vv)$.  The minimum distance of a code $C$ over $\Z_2$
is $d(C)=\min\{d_H(\uu,\vv):\uu,\vv\in C,\ \uu\not=\vv\}$.

In \cite{Sole}, a Gray map from $\Z_4$ to $\Z_2^2$ is defined as
$\phi(0)=(0,0)$, $\phi(1)=(0,1)$, $\phi(2)=(1,1)$ and $\phi(3)=(1,0)$.  There
exist different generalizations of this Gray map, which go from $\Z_{2^s}$ to
$\Z_2^{2^{s-1}}$ \cite{Carlet,Codes2k,dougherty,Nechaev,Krotov:2007}.  In this
paper, we focus on Carlet's Gray map \cite{Carlet}, from $\Z_{2^s}$ to
$\Z_2^{2^{s-1}}$, which is a particular case of the one given in
\cite{Krotov:2007,ShiKrotov2019} satisfying
$\sum\lambda_i\phi_s(2^i)=\phi_s(\sum\lambda_i2^i)$ \cite{KernelZ2s}.
Specifically,
\begin{equation}\label{eq:Carlet}
\phi_s(u)=(u_{s-1},u_{s-1},\dots,u_{s-1})+(u_0,\dots,u_{s-2})Y_{s-1},
\end{equation}
where $u\in\Z_{2^s}$; $[u_0,u_1,\dots,u_{s-1}]_2$ is the binary expansion of
$u$, that is, $u=\sum_{i=0}^{s-1}u_i2^i$ with $u_i\in\{0,1\}$; and $Y_{s-1}$
is a matrix of size $(s-1)\times2^{s-1}$ whose columns are all the vectors of
$\Z_2^{s-1}$, ordered in ascending order by reading them as binary expansions
of the elements of $\Z_{2^{s-1}}$.  Note that $\phi_1$ is the identity map and
that $\phi_2$ is the Gray map $\phi$ given above.  We define
$\Phi_s:\Z_{2^s}^n\rightarrow\Z_2^{n2^{s-1}}$ as the component-wise extended
map of $\phi_s$, and a Gray map $\Phi$ from
$\Z_2^{\alpha_1}\times\Z_4^{\alpha_2}\times\Z_8^{\alpha_3}$ to $\Z_2^N$, where
$N=\alpha_1+2\alpha_2+4\alpha_3$, by
$$
\Phi(\uu_1\mid \uu_2\mid \uu_3)=(\uu_1,\Phi_2(\uu_2),\Phi_3(\uu_3)),
$$
for any $\uu_i\in\Z_{2^i}^{\alpha_i}$, where $1\leq i\leq3$.

Let $\C\subseteq\Z_{2^s}^n$ be a $\Z_{2^s}$-additive code of length $n$.  We
say that the Gray map image of $\C$, say $C=\Phi_s(\C)$, is a
$\Z_{2^s}$-linear code of length $n2^{s-1}$.  Since $\C$ is a subgroup of
$\Z_{2^s}^n$, it is isomorphic to
$\Z_{2^s}^{t_1}\times\Z_{2^{s-1}}^{t_2}\times\dots\times\Z_2^{t_s}$, and we
say that $\C$, or equivalently $C=\Phi_s(\C)$, is of type $(n;t_1,\dots,t_s)$.
Similarly, if
$\C\subseteq\Z_2^{\alpha_1}\times\Z_4^{\alpha_2}\times\Z_8^{\alpha_3}$ is a
$\Z_2\Z_4\Z_8$-additive code, we say that its Gray map image $C=\Phi(\C)$ is a
$\Z_2\Z_4\Z_8$-linear code of length $\alpha_1+2\alpha_2+4\alpha_3$.  Since
$\C$ can be seen as a subgroup of $\Z_8^{\alpha_1+\alpha_2+\alpha_3}$, it is
isomorphic to $\Z_8^{t_1}\times\Z_4^{t_2}\times\Z_2^{t_3}$, and we say that
$\C$, or equivalently $C=\Phi(\C)$, is of type
$(\alpha_1,\alpha_2,\alpha_3;t_1,t_2,t_3)$.  Note that a $\Z_2\Z_4$-linear
code $\C$ \cite{ccsg,BookZ2Z4} can be seen as a $\Z_2\Z_4\Z_8$-linear code of
type $(\alpha_1,\alpha_2,0;0,t_2,t_3)$.  In this case, we also say that the
type of $\C$ is directly $(\alpha_1,\alpha_2;t_2,t_3)$.  Unlike linear codes
over finite fields, linear codes over rings do not have a basis, but there
exists a generator matrix for these codes having minimum number of rows.  If
$\C$ is a $\Z_2\Z_4\Z_8$-additive code of type
$(\alpha_1,\alpha_2,\alpha_3;t_1,t_2,t_3)$, then $|\C|=8^{t_1}4^{t_2}2^{t_3}$
and there exists a generator matrix with $t_1+t_2+t_3$ rows.

A binary code of length $N$, $2N$ codewords and minimum distance $N/2$ is
called a Hadamard code.  Hadamard codes can be constructed from Hadamard
matrices \cite{Key,WMcwill}.  Note that linear Hadamard codes are in fact
first order Reed--Muller codes, or equivalently, the dual of extended Hamming
codes \cite{WMcwill}.  It is also important to note that Hadamard codes are
two weight codes, which have been widely studied in
\cite{ShiTwoHomWeight,TwoWeightSole}.  Most Hadamard codes are, however,
nonlinear, and their classification is still an open problem.  One fruitful
way of attacking it is to realise some of them as Gray map images of additive
codes over rings or over mixed alphabets and then to decide which of the
resulting codes are equivalent, each such result giving a partial
classification of nonlinear Hadamard codes.  The $\Z_4$ representation of the
Kerdock, Preparata and Goethals codes \cite{Sole} was the starting point of
this point of view, and codes over $\Z_{p^s}$ go back to Blake \cite{Blake}
and Shankar \cite{Shankar}.

From a more practical point of view, since Hadamard codes are optimal and have
a high correction capability, they appear in different aspects related to the
transmission of information, such as in digital communication with satellites
\cite{H07}, in CDMA phones to modulate the transmission of information and
minimize interference with other transmissions \cite{SBT98} and, in general,
in different OCDMA multiple access systems to allow access to multiple users
asynchronously and simultaneously \cite{HYT04}.  Other applications are found
in cryptography \cite{Nyb91} or in information hiding (steganography and
watermarking) \cite{YLL03}.  See \cite{H07} for more applications in other
fields.

Two structural properties of codes over $\Z_2$ are the rank and the dimension
of the kernel.  The rank of a code $C$ over $\Z_2$ is simply the dimension of
the linear span, $\langle C\rangle$, of $C$.  The kernel of a code $C$ over
$\Z_2$ is defined as $\K(C)=\{\xx\in\Z_2^N:\xx+C=C\}$ \cite{BGH83}.  If the
all-zero vector belongs to $C$, then $\K(C)$ is a linear subcode of $C$.  Note
also that if $C$ is linear, then $\K(C)=C=\langle C\rangle$.  We denote the
rank of $C$ as $\rank(C)$ and the dimension of the kernel as $\kernel(C)$.
Both are equivalence invariants, so two codes with different pairs
$(r,k)=(\rank(C),\kernel(C))$ are nonequivalent.  The converse is false in
general, and it is precisely the failure of the converse that makes
classification results difficult.  The purpose of the present paper is
to repair that failure, for one concrete family of Hadamard codes, by means of
a new invariant.

The $\Z_{2^s}$-additive codes such that after the Gray map $\Phi_s$ give
Hadamard codes are called $\Z_{2^s}$-additive Hadamard codes and the
corresponding images are called $\Z_{2^s}$-linear Hadamard codes.  Similarly,
the $\Z_2\Z_4\Z_8$-additive codes such that after the Gray map $\Phi$ give
Hadamard codes are called $\Z_2\Z_4\Z_8$-additive Hadamard codes and the
corresponding images are called $\Z_2\Z_4\Z_8$-linear Hadamard codes.  It is
known that $\Z_4$-linear Hadamard codes, that is, $\Z_2\Z_4$-linear Hadamard
codes with $\alpha_1=0$, and $\Z_2\Z_4$-linear Hadamard codes with
$\alpha_1\not=0$ can be classified by using either the rank or the dimension
of the kernel \cite{Kro:2001:Z4_Had_Perf,PRV06}.  Moreover, in \cite{KV2015},
it is shown that each $\Z_4$-linear Hadamard code is equivalent to a
$\Z_2\Z_4$-linear Hadamard code with $\alpha_1\not=0$.  Later, in
\cite{KernelZ2s,HadamardZps,EquivZ2s,ZpsEquivalance}, a recursive construction
for $\Z_{p^s}$-linear Hadamard codes, with $p$ prime, is described, the
linearity is established, and a partial classification by using the dimension
of the kernel is obtained, giving the exact amount of nonequivalent such codes
for some parameters.  In \cite{fernandez2019mathbb}, a complete classification
of $\Z_8$-linear Hadamard codes by using the rank and dimension of the kernel
is provided, giving the exact amount of nonequivalent such codes.  For
$s\geq4$, however, no formula for the rank is known, and therefore that
invariant is not even available there.  Finally, a recursive construction of
$\Z_2\Z_4\Z_8$-additive Hadamard codes $\cH^{t_1,t_2,t_3}$, with all the
$\alpha_i$ nonzero, was given in \cite{Z2Z4Z8Construction}, their linearity
and kernel were determined in \cite{Z2Z4Z8Linearity} and their rank in
\cite{Z2Z4Z8Rank}. We write $H^{t_1,t_2,t_3}=\Phi(\cH^{t_1,t_2,t_3})$ for the
corresponding binary codes. For any $t\geq2$, the full classification of
$\Z_p\Z_{p^2}$-linear generalized Hadamard codes of length $p^t$, with $\alpha_1\not=0$,
$\alpha_2\not=0$, and $p\geq3$ prime, is given in
\cite{ZpZp2Construction,ZpZp2Classification}, by using just the dimension of
the kernel.

That last family is the subject of the present paper.  The pair $(r,k)$ does
not classify it: two infinite families of pairs of distinct types share the
length, the rank and the dimension of the kernel, namely
$$
\bigl(H^{1,2,t-6},H^{2,0,t-5}\bigr),\ t\geq7,
\qquad\text{and}\qquad
\bigl(H^{2,2,t-9},H^{3,0,t-8}\bigr),\ t\geq10.
$$
Until now, each individual pair had to be separated by a computer equivalence
test, and this had been carried out only for the lengths $2^t$ with
$3\leq t\leq11$ \cite{Z2Z4Z8Linearity}.  Proving the nonequivalences for every
length was stated as further research in \cite[\S 6]{Z2Z4Z8Linearity}.

The rank and the dimension of the kernel are global invariants, and
neither of them records where in the coordinate set the corresponding
vectors live. The kernel, by contrast, carries positional information
intrinsically.  Two binary coordinates are declared equivalent when every
kernel word takes the same value in both of them, and the resulting partition
of the coordinate set into \emph{kernel blocks} is canonically attached to the
code.  On each block, we puncture the linear span and record the dimension that
survives, and the multiset
$$
\cR(C)=\{\!\{\dim\langle C\rangle|_B:\ B\text{ a kernel block of }C\}\!\}
$$
is then an equivalence invariant, which we call the \emph{kernel-block rank
profile}.

Our main result is the explicit computation of $\cR(H^{t_1,t_2,t_3})$, in
Theorem \ref{thm:profile}.  The kernel partition has $2^{t_1+t_2+t_3-1}$
blocks, all of the same size $2^{2t_1+t_2}$, and the profile takes at most the
two values
$$
\rho_0=t_2+\binom{t_1+2}{2}
\qquad\text{and}\qquad
\rho_1=t_2+2+\binom{t_1+1}{2},
$$
whose difference is exactly $t_1-1$.  The profile therefore determines $t_1$,
and the length together with the dimension of the kernel then determines $t_2$
and $t_3$.  This gives the complete classification inside the family, Theorem
\ref{thm:complete-internal-classification}: two codes of the family with the
same length are equivalent if and only if their types coincide.  Counting the
admissible types, the number of pairwise nonequivalent such codes of length
$2^t$ is $\lfloor(t^2+6)/12\rfloor$ for every $t\geq3$, and both families of
collisions are separated uniformly, at every length and without any computer
computation.

The classification obtained here is internal: it compares the codes
$H^{t_1,t_2,t_3}$ with one another.  Whether such a code can be equivalent to
a $\Z_4$-linear, a $\Z_2\Z_4$-linear or a $\Z_{2^s}$-linear Hadamard code of
the same length is not addressed here; it requires the profiles of those
families, and will be taken up separately.

The paper is organised as follows.  Section \ref{sec:preliminaries} fixes the
vocabulary of generator rows, additive coordinates and binary coordinates, and
recalls what we use from \cite{Z2Z4Z8Construction,Z2Z4Z8Linearity,Z2Z4Z8Rank}.
Section \ref{sec:profile} introduces the kernel-block rank profile, proves its
invariance under code equivalence, and computes it for the whole family.
Section \ref{sec:internal-classification} draws the classification inside the
family and counts the equivalence classes, and Section \ref{sec:conclusions}
summarises the results and indicates further research.
%%%%%%%%%%%%%%%%%%%%%%%%%%%%%%%%%%%%%%%%%%%%%%%%%%%%%%%%%%%%%%%%%%%%%%%%%%%%%
\section{Preliminaries}\label{sec:preliminaries}
%%%%%%%%%%%%%%%%%%%%%%%%%%%%%%%%%%%%%%%%%%%%%%%%%%%%%%%%%%%%%%%%%%%%%%%%%%%%%

This section fixes the notation and collects the results that are used later.
Apart from Subsection \ref{subsec:puncturing} and Lemma \ref{lem:mobius},
which are proved here, everything in this section is taken from
\cite{Z2Z4Z8Construction,Z2Z4Z8Linearity,Z2Z4Z8Rank}.

%%%%%%%%%%%%%%%%%%%%%%%%%%%%%%%%%%%%%%%%%%%%%%%%%%%%%%%%%%%%%%%%%%%%%%%%%%%%%
\subsection{Puncturing and equivalence}\label{subsec:puncturing}
%%%%%%%%%%%%%%%%%%%%%%%%%%%%%%%%%%%%%%%%%%%%%%%%%%%%%%%%%%%%%%%%%%%%%%%%%%%%%

Throughout, a binary code $C\subseteq\Z_2^N$ is not assumed to be a vector
space, and we recall from Section \ref{sec:intro} that $C$ is linear exactly
when $\K(C)=C=\langle C\rangle$, that is, exactly when
$\rank(C)=\kernel(C)$.

We will constantly restrict codes and linear spaces to selected coordinate
positions.  For $\xx\in\Z_2^N$ and $B\subseteq\{1,\dots,N\}$, we write
$\xx|_B$ for the vector obtained from $\xx$ by deleting the coordinates whose
index lies outside $B$, and, for a set $D\subseteq\Z_2^N$, we put
$D|_B=\{\xx|_B:\xx\in D\}$; this operation is called puncturing outside $B$.
Since deleting coordinates is a $\Z_2$-linear map, $D|_B$ is a linear space whenever
$D$ is one, and
\begin{equation}\label{eq:span-puncture}
\langle D\rangle|_B=\langle D|_B\rangle
\qquad\text{for every } D\subseteq\Z_2^N.
\end{equation} The integer $\dim\langle D\rangle|_B$ measures how much of the rank of
$D$ remains visible on the positions of $B$, and it is the quantity on which
the invariant of Section \ref{sec:profile} is built.

The next elementary lemma records what a code equivalence does to the span and
to the kernel.  We include its proof because the identities themselves, and not only the equality of the two dimensions, are what we will need in Section \ref{sec:profile}.

\begin{lemma}\label{lem:rank-kernel-invariance}
Let $C,D\subseteq\Z_2^N$ be binary codes with $\zero\in C$ and $\zero\in D$,
and suppose that $D=\zz+\pi(C)$ for some coordinate permutation $\pi$ and some
$\zz\in\Z_2^N$.  Then, $\langle D\rangle=\pi(\langle C\rangle)$ and
$\K(D)=\pi(\K(C))$.  In particular, $\rank(D)=\rank(C)$ and
$\kernel(D)=\kernel(C)$.
\end{lemma}

\begin{proof}
Since $\zero\in D$ and $D=\zz+\pi(C)$, there is $\uu\in C$ with
$\zz=\pi(\uu)$.  Every element of $D$ is $\zz+\pi(\vv)$ with $\vv\in C$, and
$$
\zz+\pi(\vv)=\pi(\uu)+\pi(\vv)=\pi(\uu+\vv)\in\pi(\langle C\rangle),
$$
so $\langle D\rangle\subseteq\pi(\langle C\rangle)$.  Conversely, $\zz$ and
$\zz+\pi(\vv)$ both belong to $D$, and their sum is $\pi(\vv)$, so
$\pi(C)\subseteq\langle D\rangle$ and hence
$\pi(\langle C\rangle)\subseteq\langle D\rangle$.  This proves the first
equality, and $\rank(D)=\rank(C)$, a permutation of coordinates being a
linear isomorphism of $\Z_2^N$.

For the kernel, let $\xx\in\K(C)$.  Then,
$$
\pi(\xx)+D=\pi(\xx)+\zz+\pi(C)=\zz+\pi(\xx+C)=\zz+\pi(C)=D,
$$
so $\pi(\K(C))\subseteq\K(D)$.  Since $C=\pi^{-1}(\zz)+\pi^{-1}(D)$ is a
relation of the same shape, the inclusion just proved, applied to the pair
$(D,C)$, gives $\K(D)\subseteq\pi(\K(C))$.  Taking dimensions gives
$\kernel(D)=\kernel(C)$.
\end{proof}

%%%%%%%%%%%%%%%%%%%%%%%%%%%%%%%%%%%%%%%%%%%%%%%%%%%%%%%%%%%%%%%%%%%%%%%%%%%%%
\subsection{Generator rows, additive coordinates and binary coordinates}
\label{subsec:rows-coordinates}
%%%%%%%%%%%%%%%%%%%%%%%%%%%%%%%%%%%%%%%%%%%%%%%%%%%%%%%%%%%%%%%%%%%%%%%%%%%%%

Five different objects occur in the computations below, and it is important
to keep them apart: \emph{generator rows},
\emph{additive coordinates}, \emph{message coefficients}, \emph{binary message
digits} and \emph{binary coordinates}.  We recall the vocabulary of
\cite{Z2Z4Z8Rank}.

Let $A=(A_1\mid A_2\mid A_3)$ be a generator matrix of a
$\Z_2\Z_4\Z_8$-additive code $\C$ of type
$(\alpha_1,\alpha_2,\alpha_3;t_1,t_2,t_3)$.  The rows of $A$ are the
\emph{generator rows}, and the positions of $\C$ are its \emph{additive
coordinates}.  An additive coordinate lies in the $\Z_2$ part, in the $\Z_4$
part or in the $\Z_8$ part according to whether the column of $A$ at that
position belongs to $A_1$, to $A_2$ or to $A_3$.  That column lists the values
taken there by the generator rows, and we identify the additive coordinate
with it whenever we speak of its entries.  Choosing generators
$\ww_1,\dots,\ww_{t_1}$, $\vv_1,\dots,\vv_{t_2}$ and $\zz_1,\dots,\zz_{t_3}$
of order $8$, $4$ and $2$ respectively, every additive codeword $\uu\in\C$ is
written in exactly one way as
\begin{equation}\label{eq:additive-message}
\uu=\sum_{i=1}^{t_1}x_i\ww_i+\sum_{j=1}^{t_2}y_j\vv_j
      +\sum_{k=1}^{t_3}\epsilon_k\zz_k,
\qquad x_i\in\Z_8,\quad y_j\in\Z_4,\quad \epsilon_k\in\Z_2.
\end{equation}
The tuple
$\xi=(x_1,\dots,x_{t_1},y_1,\dots,y_{t_2},\epsilon_1,\dots,\epsilon_{t_3})$ is
the \emph{message} of $\uu$ and its entries are the \emph{message
coefficients}; the set of all messages is
$\cM=\Z_8^{t_1}\times\Z_4^{t_2}\times\Z_2^{t_3}$, and
\eqref{eq:additive-message} is a bijection $\cM\rightarrow\C$.  Replacing each
message coefficient by its binary digits, we write
\begin{equation}\label{eq:message-digits}
x_i=a_i+2b_i+4c_i\ \ \text{in }\Z_8,
\qquad
y_j=d_j+2e_j\ \ \text{in }\Z_4,
\end{equation}
with $a_i,b_i,c_i,d_j,e_j\in\{0,1\}$, and from that point on, we regard the
$a_i,b_i,c_i$, the $d_j,e_j$ and the $\epsilon_k$ as independent variables
over $\Z_2$; they are the \emph{binary message digits}.  The same convention
applies to every ring element occurring below.  For $w\in\Z_{2^s}$, we always
write $w=\sum_{i=0}^{s-1}w_i2^i$ with $w_i\in\{0,1\}$, we call $w_0$ the
\emph{bottom digit} and $w_{s-1}$ the \emph{top digit} of $w$ and, when $s=3$,
we call $w_1$ the \emph{middle digit}.  The digits are extracted inside
$\Z_{2^s}$, where they are the ring elements $0$ and $1$, and are viewed in
$\Z_2$ from then on, because by \eqref{eq:Carlet} the Gray image of $w$ is
obtained from them by linear algebra over $\Z_2$, and because every dimension
we compute is the dimension of a $\Z_2$-span. This identification is only set-theoretic: the digits of a sum of two ring elements are not, in general, the sums of the corresponding digits, because of the carries.

Finally, a position of $\Phi(\C)$ is called a \emph{binary coordinate}: one
additive coordinate in the $\Z_2$ part produces one binary coordinate, one in
the $\Z_4$ part produces two and one in the $\Z_8$ part produces four.  The
two kinds of position are always distinguished in the same way: an additive
coordinate is named together with its part, as in ``a coordinate in the $\Z_4$
part'', whereas ``binary coordinate'', or simply ``coordinate'', always means
a position of the binary code.

For a binary coordinate $\ell$ of $\Phi(\C)$, let $f_\ell:\cM\rightarrow\Z_2$
be the map sending a message $\xi$ to the value of $\Phi(\uu(\xi))$ at that
position, where $\uu(\xi)$ is given by \eqref{eq:additive-message}. We call
$f_\ell$ the \emph{binary coordinate function} of $\ell$.  Since digit
extraction identifies $\cM$, as a set, with $\Z_2^{3t_1+2t_2+t_3}$, each
$f_\ell$ is a Boolean function of the message digits, nonlinear in general.
Two things vary independently here: the message is the variable, and the
coordinate is what selects the function.  The next lemma is what turns every
rank in this paper into the dimension of a space of Boolean functions.

\begin{lemma}[{\cite[Lemma 3.3]{Z2Z4Z8Rank}}]\label{lem:evaluation-principle}
Let $\Omega$ be a finite set, let $f_1,\dots,f_N:\Omega\rightarrow\Z_2$. Fix
an ordering $\omega_1,\dots,\omega_{|\Omega|}$ of $\Omega$, and let
$C_\Omega\subseteq\Z_2^N$ be the set of rows of the $|\Omega|\times N$ binary
matrix $\bigl(f_j(\omega_i)\bigr)_{i,j}$.  Then, for every
$B\subseteq\{1,\dots,N\}$,
\begin{equation}\label{eq:restricted-evaluation-rank}
\dim\langle C_\Omega\rangle|_B=\dim\Span\{f_j:\ j\in B\}.
\end{equation}
In particular, taking $B=\{1,\dots,N\}$ gives
$\rank(C_\Omega)=\dim\Span\{f_1,\dots,f_N\}$.
\end{lemma}

Indeed, both sides are the rank of the submatrix of $\bigl(f_j(\omega_i)\bigr)$
formed by the columns indexed by $B$.

Applied with $\Omega=\cM$, Lemma \ref{lem:evaluation-principle} says that the
punctured rank of $\Phi(\C)$ on a set $B$ of binary coordinates is the
dimension of the span of the coordinate functions attached to the positions of
$B$.  The matrix of values has $2^{\,t+1}$ rows and $2^{t}$ columns and cannot
be handled directly, whereas \eqref{eq:restricted-evaluation-rank} allows us to
compute the required punctured ranks from the corresponding coordinate
functions. We will use two further consequences of \eqref{eq:Carlet}. At a fixed
additive coordinate, the binary coordinate functions produced by the Gray map
span exactly the same space as the binary digits of that coordinate. Moreover, adding an element of order at most two before applying the Gray map is equivalent to adding its Gray image afterwards.

\begin{lemma}[{\cite[Lemma 3.1]{Z2Z4Z8Rank}}]\label{lem:gray-digits}
Let $\Omega$ be a finite set and let $u:\Omega\rightarrow\Z_{2^s}$ be any map.
For $0\leq i\leq s-1$, let $u_i:\Omega\rightarrow\Z_2$ send $\omega$ to the
$i$th binary digit of $u(\omega)$, and, for $1\leq\ell\leq 2^{s-1}$, let
$g_\ell:\Omega\rightarrow\Z_2$ send $\omega$ to the $\ell$th entry of
$\phi_s(u(\omega))$.  Then,
$$
\Span\{g_1,\dots,g_{2^{s-1}}\}=\Span\{u_0,u_1,\dots,u_{s-1}\}.
$$
\end{lemma}

\begin{lemma}[\cite{Carlet,HadamardZps}]\label{lem:gray-shift}
Let $s\geq1$ and let $\lambda,w\in\Z_{2^s}$.  Then,
$\phi_s(2^{s-1}\lambda+w)=\phi_s(2^{s-1}\lambda)+\phi_s(w)$.
\end{lemma}

The last tool of this subsection says that the span of a family of functions
depending polynomially on free binary parameters is the span of the
coefficient functions.  We use it to read a local span directly from a formula,
without choosing the parameters one at a time.

\begin{lemma}\label{lem:mobius}
Let $\Omega$ be a finite set, let $\theta=(\theta_1,\dots,\theta_n)$ be free
binary parameters and suppose that, for every $\theta\in\Z_2^n$, a function
$G_\theta:\Omega\rightarrow\Z_2$ is given by
\begin{equation}\label{eq:parameter-polynomial}
G_\theta=\sum_{I\subseteq\{1,\dots,n\}}\Bigl(\prod_{i\in I}\theta_i\Bigr)f_I,
\end{equation}
where the functions $f_I:\Omega\rightarrow\Z_2$ do not depend on $\theta$.
Then, $\Span\{G_\theta:\theta\in\Z_2^n\}=\Span\{f_I:I\subseteq\{1,\dots,n\}\}$.
\end{lemma}

\begin{proof}
One inclusion is immediate. By \eqref{eq:parameter-polynomial}, every
$G_\theta$ is a $\Z_2$-linear combination of the functions $f_I$, so
$\Span\{G_\theta\}\subseteq\Span\{f_I\}$.

For the reverse inclusion, let $\chi_J\in\Z_2^n$ be the indicator vector of a
subset $J\subseteq\{1,\dots,n\}$.  Substituting $\theta=\chi_J$ in
\eqref{eq:parameter-polynomial}, the product $\prod_{i\in I}\theta_i$ equals
$1$ exactly when $I\subseteq J$, so
\begin{equation}\label{eq:mobius-values}
G_{\chi_J}=\sum_{I\subseteq J}f_I.
\end{equation}
We claim that this relation can be inverted over $\Z_2$ as
$f_I=\sum_{J\subseteq I}G_{\chi_J}$.  Indeed, substituting
\eqref{eq:mobius-values} and exchanging the order of summation,
$$
\sum_{J\subseteq I}G_{\chi_J}
=\sum_{J\subseteq I}\ \sum_{L\subseteq J}f_L
=\sum_{L\subseteq I}\bigl|\{J:\ L\subseteq J\subseteq I\}\bigr|\,f_L
=\sum_{L\subseteq I}2^{\,|I|-|L|}f_L,
$$
since $\bigl|\{J:L\subseteq J\subseteq I\}\bigr|=2^{\,|I|-|L|}$.  Over $\Z_2$,
this coefficient vanishes unless $L=I$, so the sum reduces to $f_I$.
Hence, every $f_I$ is a $\Z_2$-linear combination of the functions $G_\theta$,
which gives $\Span\{f_I\}\subseteq\Span\{G_\theta\}$.
\end{proof}

\begin{example}\label{ex:mobius}
Take $n=2$.  Then, $G_\theta=f_\emptyset+\theta_1f_{\{1\}}
+\theta_2f_{\{2\}}+\theta_1\theta_2f_{\{1,2\}}$.  The four functions obtained
by letting $\theta$ run over $\Z_2^2$ are
$G_{(0,0)}=f_\emptyset$, $G_{(1,0)}=f_\emptyset+f_{\{1\}}$, $G_{(0,1)}=f_\emptyset+f_{\{2\}}$, and $G_{(1,1)}=f_\emptyset+f_{\{1\}}+f_{\{2\}}+f_{\{1,2\}}.$
Adding the first to the second gives $f_{\{1\}}$, adding it to the third gives
$f_{\{2\}}$, and adding all four gives $f_{\{1,2\}}$; so the four functions
$G_\theta$ span the same space as $f_\emptyset$, $f_{\{1\}}$, $f_{\{2\}}$ and
$f_{\{1,2\}}$, in agreement with Lemma \ref{lem:mobius}.
\end{example}

%%%%%%%%%%%%%%%%%%%%%%%%%%%%%%%%%%%%%%%%%%%%%%%%%%%%%%%%%%%%%%%%%%%%%%%%%%%%%
\subsection{The recursive family $H^{t_1,t_2,t_3}$}\label{subsec:family}
%%%%%%%%%%%%%%%%%%%%%%%%%%%%%%%%%%%%%%%%%%%%%%%%%%%%%%%%%%%%%%%%%%%%%%%%%%%%%

Extending the notation of Section \ref{sec:intro}, let
$\two,\mathbf{3},\dots,\mathbf{7}$ denote the vectors having the element
$2,3,\dots,7$ repeated in every coordinate, and let $t_1\geq1$, $t_2\geq0$ and
$t_3\geq1$ be integers.  The matrices $A^{t_1,t_2,t_3}$ of \cite{Z2Z4Z8Construction}, having $t_1$ rows of order $8$, $t_2$ rows of order $4$ and $t_3$ rows of order $2$, are
constructed recursively from
\begin{equation}\label{eq:A101}
A^{1,0,1}=\left(\begin{array}{cc|c|c}
1&1&2&4\\
0&1&1&1\\
\end{array}\right)
\end{equation}
by means of three steps.  If $A^{\ell-1,0,1}=(A_1\mid A_2\mid A_3)$ with
$\ell\geq2$ is already constructed, we form
\begin{equation}\label{eq:construction-order8}
A^{\ell,0,1}=\left(\begin{array}{cc|ccccc|ccccc}
A_1&A_1&M_1&A_2&A_2&A_2&A_2&M_2&A_3&A_3&\cdots&A_3\\
\zero&\one&\one&\zero&\one&\two&\mathbf{3}&\one&\zero&\one&\cdots&\mathbf{7}
\end{array}\right),
\end{equation}
where $M_1=\{\zz^T:\zz\in\{2\}\times\{0,2\}^{\ell-1}\}$ and
$M_2=\{\zz^T:\zz\in\{4\}\times\{0,2,4,6\}^{\ell-1}\}$; this step is applied
until $\ell=t_1$.  If $A^{t_1,\ell-1,1}=(A_1\mid A_2\mid A_3)$ with
$\ell\geq1$ is already constructed, we form
\begin{equation}\label{eq:construction-order4}
A^{t_1,\ell,1}=\left(\begin{array}{cc|ccccc|cccc}
A_1&A_1&M_1&A_2&A_2&A_2&A_2&A_3&A_3&A_3&A_3\\
\zero&\one&\one&\zero&\one&\two&\mathbf{3}&\zero&\two&\four&\six
\end{array}\right),
\end{equation}
where now $M_1=\{\zz^T:\zz\in\{2\}\times\{0,2\}^{t_1+\ell-1}\}$; this step is
applied until $\ell=t_2$.  Finally, if $A^{t_1,t_2,\ell-1}=(A_1\mid A_2\mid
A_3)$ with $\ell\geq2$ is already constructed, we form
\begin{equation}\label{eq:construction-order2}
A^{t_1,t_2,\ell}=\left(\begin{array}{cc|cc|cc}
A_1&A_1&A_2&A_2&A_3&A_3\\
\zero&\one&\zero&\two&\zero&\four
\end{array}\right),
\end{equation}
and this last step is applied until $\ell=t_3$. Thus, in this way, we obtain $A^{t_1,t_2,t_3}$. 

Summarizing, in order to reach $A^{t_1,t_2,t_3}$ from $A^{1,0,1}$ we first
adds $t_1-1$ rows of order $8$ by applying \eqref{eq:construction-order8}
exactly $t_1-1$ times, then add $t_2$ rows of order $4$ by applying
\eqref{eq:construction-order4} exactly $t_2$ times, and finally add $t_3-1$
rows of order $2$ by applying \eqref{eq:construction-order2} exactly $t_3-1$
times.  Note that the first row of $A^{t_1,t_2,t_3}$ is always
$(\one\mid\two\mid\four)$, and that it has order $2$; it is therefore one of
the $t_3$ rows of order two, and it will play a distinguished role
throughout.  We denote by $\cH^{t_1,t_2,t_3}$ the $\Z_2\Z_4\Z_8$-additive
code generated by $A^{t_1,t_2,t_3}$ and by
$H^{t_1,t_2,t_3}=\Phi(\cH^{t_1,t_2,t_3})$ the corresponding
$\Z_2\Z_4\Z_8$-linear code.

\begin{example}\label{ex:matricesA}
By using the constructions described in (\ref{eq:construction-order8}), (\ref{eq:construction-order4}), and (\ref{eq:construction-order2}), we obtain the following matrices $A^{2,0,1}$, $A^{1,1,1}$ and $A^{1,1,2}$, respectively, starting from $A^{1,0,1}$ given in (\ref{eq:A101}):
\begin{equation*}%\label{eq:A201}
A^{2,0,1}=\left(\begin{array}{cc|cc|cc}
11&11&22&2222&4444&44444444\\
01&01&02&1111&0246&11111111\\
00&11&11&0123&1111&01234567
\end{array}\right),
\end{equation*}
\begin{equation}\label{eq:A111}
A^{1,1,1}=
\left(\begin{array}{cc|cc|c}
 11 &11 &22&2222 &4444\\
 01&01 &02&1111 &1111\\
 00&11 &11&0123 &0246\\ 
\end{array}\right), 
\end{equation}
$$
A^{1,1,2}=
\left(\begin{array}{cc|cc|cc}
 1111 & 1111 &222222 &222222 &4444 &4444\\
 0101 & 0101 &021111 &021111 &1111 &1111\\
 0011 & 0011 &110123 &110123 &0246 &0246\\ 
 0000 &1111  &000000 &222222 &0000 &4444\\
\end{array}\right).
$$
\end{example}

For all the proofs, it is far more convenient to know directly which additive
coordinates occur than to follow the recursive construction step by step.  A vector $\zz\in\Z_{2^s}^\ell$ is \emph{primitive} if at least one of its entries is odd, and hence a unit.  If $\zz$ is primitive and $i_0=\min\{i:z_i\text{ is odd}\}$, then $\zz$ is called \emph{normalized primitive} if $z_{i_0}=1$. Normalization selects one vector from each orbit under multiplication by a unit.  It fixes only the entry at the position $i_0$, leaving the even entries before $i_0$ and all the entries after $i_0$ free.  In $\Z_4^2$, for instance, the normalized primitive vectors are $(1,0)$, $(1,1)$, $(1,2)$, $(1,3)$, $(0,1)$, and $(2,1)$.

\begin{proposition}[\cite{Z2Z4Z8Rank}]\label{prop:column-description}
Let $t_1\geq1$, $t_2\geq0$ and $t_3\geq1$.  Delete from every column of
$A^{t_1,t_2,1}$ its first entry, which is $1$, $2$ or $4$ according to
whether the additive coordinate lies in the $\Z_2$, the $\Z_4$ or the $\Z_8$
part.  Then, the sets of the remaining truncated columns are the following.
\begin{enumerate}[label=\textup{(\roman*)}]
\item The $\Z_2$ part consists of every vector of $\Z_2^{t_1+t_2}$, each one
occurring exactly once.
\item The $\Z_4$ part consists of every normalized primitive vector of
$\Z_4^{t_1+t_2}$, each one occurring exactly once.
\item The $\Z_8$ part consists of every vector
$(q_1,\dots,q_{t_1},2r_1,\dots,2r_{t_2})$, where $(q_1,\dots,q_{t_1})$ is
normalized primitive in $\Z_8^{t_1}$ and $(r_1,\dots,r_{t_2})\in\Z_4^{t_2}$ is
arbitrary, each one occurring exactly once.
\end{enumerate}
Moreover, the type of $\cH^{t_1,t_2,t_3}$ satisfies
$\alpha_1=2^{t_1+t_2+t_3-1}$, $\alpha_1+2\alpha_2=4^{t_1+t_2}2^{t_3-1}$ and
$\alpha_1+2\alpha_2+4\alpha_3=8^{t_1}4^{t_2}2^{t_3-1}$; in particular
$\alpha_1$, $\alpha_2$ and $\alpha_3$ are all nonzero, the binary length of
$H^{t_1,t_2,t_3}$ is $N=2^t$ with
\begin{equation}\label{eq:length-relation}
t+1=3t_1+2t_2+t_3,
\end{equation}
and $|H^{t_1,t_2,t_3}|=8^{t_1}4^{t_2}2^{t_3}=2^{t+1}=2N$.
\end{proposition}

Note that $\zero\in H^{t_1,t_2,t_3}$, which is what the kernel and the profile
require. The matrix $A^{1,1,1}$ of \eqref{eq:A111} illustrates the three items.  After
the first entry of each column is deleted, the $\Z_2$ part gives all four
vectors of $\Z_2^2$, the $\Z_4$ part gives the six normalized primitive
vectors of $\Z_4^2$ listed above, and the $\Z_8$ part gives the four vectors
$(1,2r)$ with $r\in\Z_4$.  Thus, the type is $(4,6,4;1,1,1)$ and the binary
length is $32=2^5$.

\begin{theorem}[\cite{Z2Z4Z8Construction,Z2Z4Z8Linearity}]
\label{thm:recalled-linearity-kernel}
Let $t_1\geq1$, $t_2\geq0$ and $t_3\geq1$ be integers.  Then,
$H^{t_1,t_2,t_3}$ is a binary Hadamard code of length $2^t$ with
$t+1=3t_1+2t_2+t_3$.  Moreover, the following hold.
\begin{enumerate}[label=\textup{(\roman*)}]
\item The code $H^{t_1,t_2,t_3}$ is linear if and only if $(t_1,t_2)=(1,0)$;
equivalently, the linear members of the family are exactly the codes
$H^{1,0,t_3}$ with $t_3\geq1$.
\item If $H^{t_1,t_2,t_3}$ is nonlinear and $\cH^{t_1,t_2,t_3}_2$ denotes the
subcode of $\cH^{t_1,t_2,t_3}$ consisting of its elements of order at most
two, then $\K(H^{t_1,t_2,t_3})=\Phi(\cH^{t_1,t_2,t_3}_2)$ and
$\kernel(H^{t_1,t_2,t_3})=t_1+t_2+t_3$.
\item If $H^{t_1,t_2,t_3}$ is nonlinear, if $\uu_1,\dots,\uu_{t_1+t_2+t_3}$
are the rows of $A^{t_1,t_2,t_3}$, in any order, and if $o_i=\ord(\uu_i)$,
then $\{\Phi(\tfrac{o_i}{2}\uu_i):1\leq i\leq t_1+t_2+t_3\}$ is a basis of
$\K(H^{t_1,t_2,t_3})$.
\end{enumerate}
\end{theorem}

%%%%%%%%%%%%%%%%%%%%%%%%%%%%%%%%%%%%%%%%%%%%%%%%%%%%%%%%%%%%%%%%%%%%%%%%%%%%%
\subsection{Message digits, coordinate digits and the rank}
\label{subsec:digits-rank}
%%%%%%%%%%%%%%%%%%%%%%%%%%%%%%%%%%%%%%%%%%%%%%%%%%%%%%%%%%%%%%%%%%%%%%%%%%%%%

Whenever we compute with the family, we work with $H^{t_1,t_2,1}$ and pass to
a general $t_3$ only at the end, by means of the duplication
\eqref{eq:construction-order2}.  We write $\ww_0=(\one\mid\two\mid\four)$ for
the distinguished row of order $2$, $\ww_1,\dots,\ww_{t_1}$ for the rows of
order $8$ and $\vv_1,\dots,\vv_{t_2}$ for the rows of order $4$ of
$A^{t_1,t_2,1}$, so every additive codeword is written uniquely as
\begin{equation}\label{eq:general-message}
\epsilon\ww_0+\sum_{i=1}^{t_1}x_i\ww_i+\sum_{j=1}^{t_2}y_j\vv_j,
\qquad \epsilon\in\Z_2,\quad x_i\in\Z_8,\quad y_j\in\Z_4,
\end{equation}
with the digits of $x_i$ and $y_j$ as in \eqref{eq:message-digits}.  Every
Boolean function has a unique algebraic normal form, so distinct square-free
monomials in independent binary variables are linearly independent; this fact
justifies every dimension count below.  For a list $\zz=(z_1,\dots,z_n)$ of
binary variables, we write $\sym_j(\zz)=\sum_{i_1<\cdots<i_j}z_{i_1}\cdots
z_{i_j}$ for its $j$th elementary symmetric Boolean polynomial.  For subsets $S\subseteq\{1,\dots,t_1\}$ and
$R\subseteq\{1,\dots,t_2\}$, we abbreviate
\begin{equation}\label{eq:symmetric-notation}
\alpha_S=\sum_{i\in S}a_i,\quad \beta_S=\sum_{i\in S}b_i,\quad
\gamma_S=\sum_{i\in S}c_i,\quad \delta_R=\sum_{j\in R}d_j,\quad
\eta_R=\sum_{j\in R}e_j,
\end{equation}
and we write $\sym_2(a_S)$, $\sym_2(b_S)$, $\sym_2(d_R)$ and $\sym_4(a_S)$ for
the elementary symmetric polynomials of the corresponding sublists; for
instance, $\sym_2(a_S)=\sum_{\{i,k\}\subseteq S}a_ia_k$ and
$\sym_4(a_S)=\sum_{\{i,j,k,l\}\subseteq S}a_ia_ja_ka_l$.  

The next three statements are taken from \cite{Z2Z4Z8Rank} and rewritten in
the notation \eqref{eq:symmetric-notation}.  The first one is the rule
governing the carries of a sum of bits, and the other two describe the binary
digits carried by an additive coordinate in the $\Z_4$ and in the $\Z_8$ part.

\begin{lemma}[{\cite[Lemma 3.6]{Z2Z4Z8Rank}}]\label{lem:carries}
Let $z_1,\dots,z_n\in\Z_2$ and let $w=z_1+\cdots+z_n$ be their ordinary
integer sum.  Then, for every $r\geq0$, the $r$th binary digit of $w$ equals
$\sym_{2^r}(z_1,\dots,z_n)$.  In particular, the bottom digit, the carry to
the next digit and the direct carry to the third digit are $\sym_1$, $\sym_2$
and $\sym_4$, respectively.
\end{lemma}

\begin{lemma}[{\cite[proof of Lemma 3.4]{Z2Z4Z8Rank}}]\label{lem:z4-digits}
Consider the additive coordinate of $A^{t_1,t_2,1}$ in the $\Z_4$ part whose
column is $(2,p_1,\dots,p_{t_1},s_1,\dots,s_{t_2})^T$. Write the binary
expansions $p_i=\lambda_i+2\mu_i$ and $s_j=\lambda'_j+2\mu'_j$ in $\Z_4$, and
put $S=\{i:\lambda_i=1\}$, $T=\{i:\mu_i=1\}$, $R=\{j:\lambda'_j=1\}$ and
$R'=\{j:\mu'_j=1\}$.  Then, its bottom and top binary digits are
\begin{equation}\label{eq:z4-digits}
P=\alpha_S+\delta_R,
\qquad
h=\epsilon+\beta_S+\alpha_T+\eta_R+\delta_{R'}
  +\sym_2(a_S)+\alpha_S\delta_R+\sym_2(d_R),
\end{equation}
and, by Lemma \ref{lem:gray-digits}, the two binary coordinates it produces
span the same space of Boolean functions as $P$ and $h$.  Moreover, as the
additive coordinate runs over the $\Z_4$ part, the quadruple $(S,T,R,R')$ runs
exactly over those quadruples for which $(S,R)\not=(\emptyset,\emptyset)$ and
which satisfy one normalization constraint: $\min S\notin T$ when
$S\not=\emptyset$, and $\min R\notin R'$ when $S=\emptyset$.
\end{lemma}

\begin{lemma}[{\cite[Lemmas 3.5 and 3.7]{Z2Z4Z8Rank}}]\label{lem:z8-digits}
Consider the additive coordinate of $A^{t_1,t_2,1}$ in the $\Z_8$ part whose
column is
\begin{equation}\label{eq:z8-column}
(4,q_1,\dots,q_{t_1},2r_1,\dots,2r_{t_2})^T,
\end{equation}
write the binary expansions
\begin{equation}\label{eq:column-digits}
q_i=\lambda_i+2\mu_i+4\nu_i\ \text{ in }\Z_8,
\qquad
r_j=\lambda'_j+2\mu'_j\ \text{ in }\Z_4,
\end{equation}
and put $S=\{i:\lambda_i=1\}$, $T=\{i:\mu_i=1\}$, $U=\{i:\nu_i=1\}$,
$R=\{j:\lambda'_j=1\}$ and $R'=\{j:\mu'_j=1\}$.  Then, $S\not=\emptyset$ and,
writing
\begin{equation}\label{eq:Q0F0}
\begin{aligned}
Q_0&=\beta_S+\delta_R+\sym_2(a_S),\\
F_0&=\epsilon+\gamma_S+\eta_R+\sym_2(b_S)+\sym_2(d_R)+\sym_4(a_S)
    +\beta_S\delta_R+\bigl(\beta_S+\delta_R\bigr)\sym_2(a_S),
\end{aligned}
\end{equation}
the bottom, the middle and the top binary digit of that coordinate are
\begin{align}
P&=\alpha_S,\label{eq:bottom-digit}\\
Q&=Q_0+\alpha_T,\label{eq:middle-digit}\\
F&=F_0+\alpha_U+\delta_{R'}+\beta_T+\sym_2(a_T)+\alpha_T\,Q_0,
  \label{eq:top-digit}
\end{align}
and, by Lemma \ref{lem:gray-digits}, the four binary coordinates it produces
span the same space of Boolean functions as $P$, $Q$ and $F$.  Moreover,
writing $i_0=\min S$, as the additive coordinate runs over the $\Z_8$ part the
set $S$ runs over all the nonempty subsets of $\{1,\dots,t_1\}$, the pair
$(T,U)$ runs over all the pairs of subsets of
$\{1,\dots,t_1\}\setminus\{i_0\}$, and $(R,R')$ runs over all the pairs of
subsets of $\{1,\dots,t_2\}$.
\end{lemma}

Both lemmas are obtained by adding the entries of the column with the
coefficients of the message and determining the binary digits of the result
using Lemma \ref{lem:carries}. The symmetric polynomials $\sym_2$ and $\sym_4$
appearing in \eqref{eq:z4-digits} and in \eqref{eq:Q0F0} are exactly the
carries that this produces.  In both of them, $P$ denotes the bottom digit of
the coordinate under consideration.

The functions $Q_0$ and $F_0$ are the values of $Q$ and $F$ for
$T=U=R'=\emptyset$, so \eqref{eq:middle-digit} and \eqref{eq:top-digit}
display exactly how the two higher digits depend on the free digits of the
additive coordinate.  This is the form in which we will use Lemma
\ref{lem:z8-digits}, together with Lemma \ref{lem:mobius}.

The additive coordinates in the $\Z_2$ part need no separate lemma.  By item
(i) of Proposition \ref{prop:column-description}, such a coordinate has column
$(1,\lambda,\lambda')^T$ with $\lambda\in\Z_2^{t_1}$ and
$\lambda'\in\Z_2^{t_2}$ arbitrary, so, in the notation of Lemma
\ref{lem:z4-digits}, its single binary coordinate function is
$\epsilon+\alpha_S+\delta_R$.  Letting $S$ and $R$ vary, the $\Z_2$ part
therefore contributes exactly the space
\begin{equation}\label{eq:z2-part-span}
\cV_2=\Span\{\epsilon,a_1,\dots,a_{t_1},d_1,\dots,d_{t_2}\}.
\end{equation}

\begin{example}
Let us read the three lemmas on the matrix $A^{1,1,1}$ of \eqref{eq:A111},
where $t_1=t_2=1$. We write $x_1=a_1+2b_1+4c_1\in\Z_8$ and
$y_1=d_1+2e_1\in\Z_4$ for the two message coefficients, and $\epsilon\in\Z_2$
for the coefficient of $\ww_0$.  Consider the additive coordinate of the
$\Z_4$ part whose column is $(2,1,2)^T$.  Here, $p_1=1$ and $s_1=2$, so
$\lambda_1=1$, $\mu_1=0$, $\lambda'_1=0$ and $\mu'_1=1$; hence $S=\{1\}$,
$T=\emptyset$, $R=\emptyset$ and $R'=\{1\}$.  Substituting in
\eqref{eq:z4-digits}, the two digit functions are
$$
P=a_1
\qquad\text{and}\qquad
h=\epsilon+b_1+d_1.
$$
Consider next the additive coordinate of the $\Z_8$ part whose column is
$(4,1,2)^T$, that is, $q_1=1$ and $r_1=1$.  Then, $\lambda_1=1$,
$\mu_1=\nu_1=0$ and $\lambda'_1=1$, $\mu'_1=0$, so $S=\{1\}$,
$T=U=R'=\emptyset$ and $R=\{1\}$.  By \eqref{eq:Q0F0}, we get
$Q_0=b_1+d_1$ and $F_0=\epsilon+c_1+e_1+b_1d_1$, and \eqref{eq:bottom-digit},
\eqref{eq:middle-digit} and \eqref{eq:top-digit} give
$$
P=a_1,\qquad Q=b_1+d_1,\qquad F=\epsilon+c_1+e_1+b_1d_1.
$$
The coordinate of the $\Z_2$ part with column $(1,1,0)^T$ has the single
coordinate function $\epsilon+a_1$.
\end{example}
We finally recall the rank formula, the behaviour of the rank under the
duplication \eqref{eq:construction-order2}, which is a Plotkin extension, and
the exact list of the pairs that the rank and the dimension of the kernel fail
to separate.

\begin{theorem}[{\cite[Theorem 3.15]{Z2Z4Z8Rank}}]\label{thm:rank-main}
Let $t_1\geq1$, $t_2\geq0$ and $t_3\geq1$ be integers.  Then,
\begin{equation}\label{eq:rank-binomial}
\rank(H^{t_1,t_2,t_3})=t_3-1+4t_1+4\binom{t_1}{2}+2\binom{t_1}{3}
+\binom{t_1}{4}+t_2\binom{t_1+2}{2}+\binom{t_2+1}{2}.
\end{equation}
In particular, $\rank(H^{t_1,t_2,t_3+1})=1+\rank(H^{t_1,t_2,t_3})$.
\end{theorem}

\begin{proposition}[{\cite[Theorem 4.10]{Z2Z4Z8Rank}}]\label{prop:all-collisions}
Let $H^{t_1,t_2,t_3}$ and $H^{t'_1,t'_2,t'_3}$ be two nonlinear members of the
family with distinct types, having the same length $2^t$, the same rank and
the same dimension of the kernel.  Then, after interchanging them if
necessary, either 
$$
(t_1,t_2,t_3)=(1,2,t-6)\ \text{ and }\ (t'_1,t'_2,t'_3)=(2,0,t-5),
\qquad t\geq7, 
$$
or
$$
(t_1,t_2,t_3)=(2,2,t-9)\ \text{ and }\ (t'_1,t'_2,t'_3)=(3,0,t-8).
\qquad t\geq10.
$$
Conversely, each of these two pairs does have equal length,
equal rank and equal dimension of the kernel.
\end{proposition}

%%%%%%%%%%%%%%%%%%%%%%%%%%%%%%%%%%%%%%%%%%%%%%%%%%%%%%%%%%%%%%%%%%%%%%%%%%%%%
\section{The kernel-block rank profile}\label{sec:profile}
%%%%%%%%%%%%%%%%%%%%%%%%%%%%%%%%%%%%%%%%%%%%%%%%%%%%%%%%%%%%%%%%%%%%%%%%%%%%%

This section introduces the invariant announced in Section \ref{sec:intro} and
computes it for the whole family.  The idea is simple.  The kernel of a code
containing the zero word is a linear subcode, so it can be used to compare
coordinate positions: two binary coordinates are declared equivalent when
every kernel word takes the same value in both of them.  The resulting
partition of the coordinate set is intrinsic, and its classes will be called
kernel blocks.  On each block we puncture the binary span and record its
dimension, and the multiset of these local dimensions is the kernel-block rank
profile.

%%%%%%%%%%%%%%%%%%%%%%%%%%%%%%%%%%%%%%%%%%%%%%%%%%%%%%%%%%%%%%%%%%%%%%%%%%%%%
\subsection{Definition and invariance}\label{subsec:profile-definition}
%%%%%%%%%%%%%%%%%%%%%%%%%%%%%%%%%%%%%%%%%%%%%%%%%%%%%%%%%%%%%%%%%%%%%%%%%%%%%

\begin{definition}\label{def:kernel-block-profile}
Let $C\subseteq\Z_2^N$ be a binary code with $\zero\in C$ and put $K=\K(C)$.
Define a relation on the set $\{1,\dots,N\}$ of binary coordinates by
\begin{equation}\label{eq:kernel-relation}
i\sim_K j\quad\Longleftrightarrow\quad
x_i=x_j\ \text{ for every }\xx=(x_1,\dots,x_N)\in K.
\end{equation}
This is an equivalence relation; its classes are called the \emph{kernel
blocks} of $C$, and the corresponding partition of $\{1,\dots,N\}$ is denoted
by $\cB(K)$.  For a block $B\in\cB(K)$, we call
\begin{equation}\label{eq:local-rank}
\rho_C(B)=\dim\langle C\rangle|_B
\end{equation}
the \emph{local rank} of $C$ on $B$, which by \eqref{eq:span-puncture} is also
the dimension of $\langle C|_B\rangle$.
The multiset
\begin{equation}\label{eq:profile}
\cR(C)=\{\!\{\rho_C(B):\ B\in\cB(K)\}\!\}
\end{equation}
is called the \emph{kernel-block rank profile} of $C$.
\end{definition}

The relation \eqref{eq:kernel-relation} is reflexive, symmetric and transitive, so $\cB(K)$ is indeed a partition of the whole
coordinate set into nonempty and pairwise disjoint classes.  The word
``equivalence'' does not imply that the classes have the same cardinality; for
a general code they need not.  We emphasise that \eqref{eq:profile} is a
multiset and not a set: the blocks carry no preferred order, but the number of
blocks having each local rank is part of the invariant.  We write
$\{\!\{\rho^{[M]}\}\!\}$ to indicate that the value $\rho$ occurs with
multiplicity $M$.  Finally, since $K$ is a linear space, condition
\eqref{eq:kernel-relation} may be checked on any basis of $K$, and we will
always do so.  
\begin{example}
Let $C$ be the binary linear code of length $5$ generated by $11000$ and
$00111$, so $\K(C)=C=\langle C\rangle$.  In every codeword, the first two
coordinates take the same value and the last three coordinates take the same
value, while the generator $11000$ separates the two groups.  Hence
$\cB(K)=\{\{1,2\},\{3,4,5\}\}$.  The two blocks have different sizes,
and on either block the punctured code is a one-dimensional repetition code.
Therefore, $\cR(C)=\{\!\{1^{[2]}\}\!\}$.  Replacing this multiset by the ordinary set
$\{1\}$ would lose the information that there are two kernel blocks.
\end{example}

\begin{lemma}\label{lem:profile-invariant}
Let $C,D\subseteq\Z_2^N$ be equivalent binary codes with $\zero\in C$ and
$\zero\in D$.  Then, $\cR(C)=\cR(D)$.  More precisely, if $D=\zz+\pi(C)$ for a
coordinate permutation $\pi$ and a vector $\zz\in\Z_2^N$, then $\pi$ maps
$\cB(\K(C))$ bijectively onto $\cB(\K(D))$ and $\rho_D(\pi(B))=\rho_C(B)$ for
every $B\in\cB(\K(C))$.
\end{lemma}

\begin{proof}
By Lemma \ref{lem:rank-kernel-invariance}, we have $\K(D)=\pi(\K(C))$ and
$\langle D\rangle=\pi(\langle C\rangle)$.  A vector $\yy=\pi(\xx)$ satisfies
$y_{\pi(i)}=x_i$ for every $i$, so
$$
x_i=x_j\ \text{ for every }\xx\in\K(C)
\quad\Longleftrightarrow\quad
y_{\pi(i)}=y_{\pi(j)}\ \text{ for every }\yy\in\K(D);
$$
that is, $i\sim_{\K(C)}j$ if and only if $\pi(i)\sim_{\K(D)}\pi(j)$.  Hence,
$\pi$ carries the blocks of $C$ bijectively onto the blocks of $D$. Moreover, $\langle D\rangle|_{\pi(B)}=\pi(\langle C\rangle)|_{\pi(B)}$ is obtained from $\langle C\rangle|_B$ by a permutation of coordinates, so the two spaces have the same dimension and $\rho_D(\pi(B))=\rho_C(B)$.  The two multisets \eqref{eq:profile} therefore coincide.
\end{proof}

In the linear case, the profile carries nothing.  Indeed, if $C$ is linear,
then $\K(C)=C$, so all the codewords are constant on every block and the
punctured span there is contained in $\{\zero,\one\}$.  For a linear Hadamard code of length $2^t$, all the blocks are singletons, since the columns of a generator matrix are all distinct, and the all-one word belongs to the code; hence $\cR(C)=\{\!\{1^{[2^t]}\}\!\}$. Since linearity is itself an equivalence invariant, the linear member $H^{1,0,t-2}$ of the family never has to be compared with the nonlinear ones, and from now on we always assume that
$H^{t_1,t_2,t_3}$ is nonlinear, that is, that $(t_1,t_2)\not=(1,0)$ by item
(i) of Theorem \ref{thm:recalled-linearity-kernel}.

%%%%%%%%%%%%%%%%%%%%%%%%%%%%%%%%%%%%%%%%%%%%%%%%%%%%%%%%%%%%%%%%%%%%%%%%%%%%%
\subsection{The kernel labels and the blocks}\label{subsec:labels}
%%%%%%%%%%%%%%%%%%%%%%%%%%%%%%%%%%%%%%%%%%%%%%%%%%%%%%%%%%%%%%%%%%%%%%%%%%%%%

Let $H^{t_1,t_2,1}$ be nonlinear and keep the notation of Subsection
\ref{subsec:digits-rank}, in which $\ww_0=(\one\mid\two\mid\four)$ is the
distinguished row of order $2$, the rows $\ww_1,\dots,\ww_{t_1}$ have order
$8$ and the rows $\vv_1,\dots,\vv_{t_2}$ have order $4$.  By item (iii) of
Theorem \ref{thm:recalled-linearity-kernel}, a basis of $\K(H^{t_1,t_2,1})$ is
obtained by halving each row, that is,
\begin{equation}\label{eq:kernel-basis}
\Phi(\ww_0),\qquad \Phi(4\ww_i)\ (1\leq i\leq t_1),\qquad
\Phi(2\vv_j)\ (1\leq j\leq t_2).
\end{equation}
Since $\phi_2(2)=(1,1)$ and $\phi_3(4)=(1,1,1,1)$, the first of these vectors
is
$$
\Phi(\ww_0)=\Phi(\one\mid\two\mid\four)=\one,
$$
which takes the value $1$ in every binary coordinate and therefore does not
separate any two of them.  The remaining $t_1+t_2$ basis vectors do the whole
work, and we attach to every binary coordinate $\ell$ the vector
\begin{equation}\label{eq:label}
\lab(\ell)=\bigl(
(\Phi(4\ww_1))_\ell,\dots,(\Phi(4\ww_{t_1}))_\ell,
(\Phi(2\vv_1))_\ell,\dots,(\Phi(2\vv_{t_2}))_\ell
\bigr)\in\Z_2^{t_1}\times\Z_2^{t_2},
\end{equation}
which we call the \emph{kernel label} of $\ell$.  By
\eqref{eq:kernel-relation}, two binary coordinates lie in the same kernel
block exactly when they have the same label.  We therefore write
\begin{equation}\label{eq:block-def}
B_{u,v}=\bigl\{\ell:\ \lab(\ell)=(u,v)\bigr\},
\qquad (u,v)\in\Z_2^{t_1}\times\Z_2^{t_2},
\end{equation}
and the nonempty sets $B_{u,v}$ are precisely the kernel blocks of
$H^{t_1,t_2,1}$.  A different basis of the kernel would replace the labels by
their images under an invertible linear map, without changing which
coordinates have equal labels; the labels are only a convenient way of
describing the intrinsic partition $\cB(K)$.

\begin{lemma}\label{lem:kernel-labels}
Let $H^{t_1,t_2,1}$ be nonlinear and write $\lab(\ell)=(u,v)$ with
$u\in\Z_2^{t_1}$ and $v\in\Z_2^{t_2}$.  Then, the labels are the following.
\begin{enumerate}[label=\textup{(\roman*)}]
\item If $\ell$ is one of the four binary coordinates produced by the additive
coordinate in the $\Z_8$ part whose column is
$(4,q_1,\dots,q_{t_1},2r_1,\dots,2r_{t_2})^T$, then
\begin{equation}\label{eq:z8-label}
\lab(\ell)=(q_1\bmod 2,\dots,q_{t_1}\bmod 2,
            r_1\bmod 2,\dots,r_{t_2}\bmod 2),
\end{equation}
the same for all four of them.  In particular, $u\not=\zero$.
\item If $\ell$ is one of the two binary coordinates produced by the additive
coordinate in the $\Z_4$ part whose column is
$(2,p_1,\dots,p_{t_1},s_1,\dots,s_{t_2})^T$, then
\begin{equation}\label{eq:z4-label}
\lab(\ell)=\bigl(\zero,\ s_1\bmod 2,\dots,s_{t_2}\bmod 2\bigr),
\end{equation}
the same for both of them.  In particular, $u=\zero$, whatever the parities of
$p_1,\dots,p_{t_1}$ are.
\item If $\ell$ is the binary coordinate produced by an additive coordinate in
the $\Z_2$ part, then $\lab(\ell)=(\zero,\zero)$.
\end{enumerate}
\end{lemma}

\begin{proof}
By \eqref{eq:label}, computing a label amounts to evaluating the kernel
generators $\Phi(4\ww_i)$ and $\Phi(2\vv_j)$ at $\ell$, and $\Phi$ acts on
each additive coordinate separately.  We treat the three parts in turn.

At an additive coordinate of the $\Z_8$ part, the row $\ww_i$ has the entry
$q_i\in\Z_8$, so the kernel generator $4\ww_i$ has the entry $4q_i$ modulo
$8$, which is $0$ when $q_i$ is even and $4$ when $q_i$
is odd.  Since $\phi_3(0)=(0,0,0,0)$ and $\phi_3(4)=(1,1,1,1)$, all four
binary coordinates produced there receive the same value $q_i\bmod2$ from
$\Phi(4\ww_i)$.  Similarly, the row $\vv_j$ has the entry $2r_j$, so the
kernel generator $2\vv_j$ has the entry $4r_j$ modulo $8$, which by the same
computation is $0$ when $r_j$ is even and $4$ when $r_j$ is odd; the four
binary coordinates therefore receive the value $r_j\bmod2$.  This is exactly
\eqref{eq:z8-label}.  Finally, by item (iii) of Proposition
\ref{prop:column-description}, the vector $(q_1,\dots,q_{t_1})$ is normalized
primitive, so at least one $q_i$ is odd and $u\not=\zero$.

We pass to the $\Z_4$ part, where the row $\ww_i$ has an entry $p_i\in\Z_4$,
so the kernel generator $4\ww_i$ has the entry $4p_i=0$; this is why the first
$t_1$ entries of the label vanish in \eqref{eq:z4-label}.  The row $\vv_j$ has
the entry $s_j\in\Z_4$, so $2\vv_j$ has the entry $2s_j$ modulo $4$, which is
$0$ when $s_j$ is even and $2$ when $s_j$ is odd.  Since $\phi_2(2)=(1,1)$,
both binary coordinates receive the value $s_j\bmod2$.

Finally, at a coordinate of the $\Z_2$ part both $4\ww_i$ and $2\vv_j$
vanish, so the label is $(\zero,\zero)$.
\end{proof}
\begin{remark}\label{rem:labels-not-parities}
Item (ii) is easy to misread: a label is not the sequence of parities of the
entries of the column, but the sequence of values taken there by the binary
kernel generators.  In the $\Z_8$ part, the two sequences happen to agree, and
$u$ is the parity vector of $(q_1,\dots,q_{t_1})$; in the $\Z_4$ part they do
not, since $4=0$ in $\Z_4$ forces $u=\zero$ whatever the parities of
$p_1,\dots,p_{t_1}$ may be.
\end{remark}

\begin{lemma}\label{lem:block-count}
Let $H^{t_1,t_2,1}$ be nonlinear.  Then, every $B_{u,v}$ is nonempty, the
kernel partition has exactly $2^{t_1+t_2}$ blocks, and every block has exactly
$2^{2t_1+t_2}$ binary coordinates.  Moreover:
\begin{enumerate}[label=\textup{(\roman*)}]
\item if $u\not=\zero$, then $B_{u,v}$ comes from the $\Z_8$ part only;
\item if $u=\zero$ and $v\not=\zero$, then $B_{\zero,v}$ comes from the $\Z_4$
part only;
\item $B_{\zero,\zero}$ consists of the whole $\Z_2$ part together with the
binary coordinates produced by the additive coordinates in the $\Z_4$ part
whose last $t_2$ entries are all even.
\end{enumerate}
\end{lemma}

\begin{proof}
Items (i), (ii) and (iii) are immediate from Lemma \ref{lem:kernel-labels}: a
binary coordinate produced in the $\Z_8$ part has $u\not=\zero$, one produced
in the $\Z_4$ part has $u=\zero$, and one produced in the $\Z_2$ part has
$(u,v)=(\zero,\zero)$.  It remains to count the coordinates carrying each
label, and we do so for each of the three shapes of label.  Throughout, we use
the description of the additive coordinates given in Proposition
\ref{prop:column-description}, and we recall from Subsection
\ref{subsec:family} that a column is normalized primitive when its first odd
entry equals $1$.

Consider first a label $(u,v)$ with $u\not=\zero$, and let $i_0$ be
the first position at which $u$ is nonzero.  By item (i), the binary
coordinates of $B_{u,v}$ are those produced by the additive coordinates of the
$\Z_8$ part whose column \eqref{eq:z8-column} satisfies $q_i\equiv u_i$ and
$r_j\equiv v_j$ modulo $2$; so we have to count those columns.  Such a column
is primitive, because $u\not=\zero$ forces at least one $q_i$ to be odd, and
its first odd entry is $q_{i_0}$ by the choice of $i_0$; normalization
therefore imposes $q_{i_0}=1$, which kills the two higher digits of $q_{i_0}$
and nothing else.  Each of the other $t_1-1$ entries $q_i\in\Z_8$ has
prescribed parity and two free higher digits, which gives $4$ choices for it,
and each $r_j\in\Z_4$ has prescribed parity and one free higher digit, which
gives $2$ choices.  There are thus $4^{t_1-1}2^{t_2}$ such additive
coordinates, each producing four binary ones, so
$$
|B_{u,v}|=4\cdot4^{t_1-1}2^{t_2}=4^{t_1}2^{t_2}=2^{2t_1+t_2}.
$$

Consider next a label with $u=\zero$ and $v\not=\zero$.  By item (ii), the
binary coordinates of $B_{\zero,v}$ are those produced by the additive
coordinates of the $\Z_4$ part whose last $t_2$ entries satisfy
$s_j\equiv v_j$ modulo $2$. The parities of the entries $p_i$ are
unrestricted, by item (ii) of Lemma \ref{lem:kernel-labels}.  Let
$\lambda\in\Z_2^{t_1}$ collect those unrestricted parities.  Once $\lambda$ is
fixed, the parity vector of the column, which is the concatenation
$(\lambda,v)$, is completely determined, and it is nonzero because
$v\not=\zero$; hence the column is primitive.  Normalization fixes the entry
at the first nonzero position of $(\lambda,v)$ to be $1$, that is, it kills
exactly one of the $t_1+t_2$ higher digits and leaves the other $t_1+t_2-1$
free.  For each of the $2^{t_1}$ choices of $\lambda$, there are therefore
$2^{t_1+t_2-1}$ such columns, and each of them produces two binary
coordinates, so
$$
|B_{\zero,v}|=2\cdot2^{t_1}\cdot2^{t_1+t_2-1}=2^{2t_1+t_2}.
$$

There remains the label $(\zero,\zero)$, for which, by item (iii), two
contributions have to be added.  The $\Z_2$ part contributes all of its
coordinates, and by item (i) of
Proposition \ref{prop:column-description}, there are $2^{t_1+t_2}$ of them,
each producing one binary coordinate.  The $\Z_4$ part contributes the columns
all of whose entries $s_j$ are even. For such a column, the parity vector is
$(\lambda,\zero)$, which has to be nonzero for the column to be primitive, so
now $\lambda\not=\zero$ and only $2^{t_1}-1$ choices of $\lambda$ are
admissible.  Counting exactly as above gives $(2^{t_1}-1)2^{t_1+t_2-1}$ such
columns, hence $(2^{t_1}-1)2^{t_1+t_2}$ binary coordinates.  Adding the two
contributions,
$$
|B_{\zero,\zero}|=2^{t_1+t_2}+(2^{t_1}-1)2^{t_1+t_2}
=2^{t_1}\cdot2^{t_1+t_2}=2^{2t_1+t_2}.
$$

Every one of the $2^{t_1+t_2}$ labels therefore occurs, so all the sets
$B_{u,v}$ are nonempty and the kernel partition has exactly $2^{t_1+t_2}$
blocks, each of $2^{2t_1+t_2}$ binary coordinates.  Multiplying the two counts
gives $2^{3t_1+2t_2}=2^t$, the length of the code.
\end{proof}

\begin{example}\label{ex:H111-blocks}
For $A^{1,1,1}$ as in \eqref{eq:A111}, we have $t_1=t_2=1$, so a label is a
pair $(u,v)$ of bits.  The four additive coordinates in the $\Z_8$ part are
$(4,1,2r)^T$ with
$r\in\Z_4$. By \eqref{eq:z8-label}, their label is $(1\bmod 2,\ r\bmod2)$, so
the two even values $r=0,2$ give the label $(\one,\zero)$ and the two odd
values $r=1,3$ give $(\one,\one)$; both blocks therefore have $8$ elements.
Among the six additive coordinates in the $\Z_4$ part, whose truncated columns are the six normalized primitive vectors of $\Z_4^2$, the two whose last entry is even are
$(2,1,0)^T$ and $(2,1,2)^T$; by \eqref{eq:z4-label} they carry the label
$(\zero,\zero)$, while the four remaining ones, namely $(2,1,1)^T$,
$(2,1,3)^T$, $(2,0,1)^T$ and $(2,2,1)^T$, carry $(\zero,\one)$.  Hence,
$B_{\zero,\one}$ has $8$ elements, and $B_{\zero,\zero}$ consists of the four
binary coordinates of the $\Z_2$ part together with the four arising from the
two $\Z_4$ coordinates just listed, again $8$ in total.  Thus, the four
blocks have $4\cdot8=32$ elements altogether, which is the length of the code.
\end{example}

%%%%%%%%%%%%%%%%%%%%%%%%%%%%%%%%%%%%%%%%%%%%%%%%%%%%%%%%%%%%%%%%%%%%%%%%%%%%%
\subsection{The local ranks}\label{subsec:local-ranks}
%%%%%%%%%%%%%%%%%%%%%%%%%%%%%%%%%%%%%%%%%%%%%%%%%%%%%%%%%%%%%%%%%%%%%%%%%%%%%

We now compute $\rho_{H^{t_1,t_2,1}}(B)$ for every block $B$.  Recall from
Lemma \ref{lem:evaluation-principle}, applied with $\Omega=\cM$, that the
punctured rank on a set $B$ of binary coordinates equals the dimension of the
span of the coordinate functions attached to the positions of $B$.  By Lemma
\ref{lem:kernel-labels}, an additive coordinate has a well defined label, so
the computation consists in fixing a label, letting all the additive
coordinates carrying it vary, and determining the resulting span of Boolean
functions.  We treat the three shapes of label separately, in the order in
which they were listed in Lemma \ref{lem:block-count}.

\begin{lemma}\label{lem:local-rank-00}
Let $H^{t_1,t_2,1}$ be nonlinear.  Then,
\begin{equation}\label{eq:rho00}
\rho_{H^{t_1,t_2,1}}(B_{\zero,\zero})=t_2+\binom{t_1+2}{2}.
\end{equation}
\end{lemma}

\begin{proof}
By item (iii) of Lemma \ref{lem:block-count}, the block $B_{\zero,\zero}$
consists of the whole $\Z_2$ part together with the binary coordinates
produced by those additive coordinates of the $\Z_4$ part whose last $t_2$
entries are even.  In the notation of Lemma \ref{lem:z4-digits}, where
$R=\{j:s_j\text{ is odd}\}$, the latter are exactly the coordinates with
$R=\emptyset$; and for them $S\not=\emptyset$, because that lemma also records
the constraint $(S,R)\not=(\emptyset,\emptyset)$.  We compute the two
contributions separately and then add them.

The $\Z_2$ part comes first.  By \eqref{eq:z2-part-span}, its coordinate
functions span exactly
$$
\cV_2=\Span\{\epsilon,a_1,\dots,a_{t_1},d_1,\dots,d_{t_2}\},
$$
a space of dimension $1+t_1+t_2$.

For the coordinates coming from the $\Z_4$ part, substituting $R=\emptyset$
into \eqref{eq:z4-digits}, the bottom and the top digit functions of such a
coordinate are
\begin{equation}\label{eq:B00-digits}
P=\alpha_S,
\qquad
h=\epsilon+\beta_S+\alpha_T+\delta_{R'}+\sym_2(a_S),
\end{equation}
where $S\not=\emptyset$, $R'$ is an arbitrary subset of
$\{1,\dots,t_2\}$, and $T$ is an arbitrary subset of
$\{1,\dots,t_1\}\setminus\{\min S\}$, by the normalization constraint recorded
in Lemma \ref{lem:z4-digits}.

Here, $P$ lies in $\cV_2$, and so do the terms $\epsilon$, $\alpha_T$ and
$\delta_{R'}$ of $h$; hence, modulo $\cV_2$,
$$h\equiv\beta_S+\sym_2(a_S)\pmod{\cV_2}.$$
The punctured span on $B_{\zero,\zero}$ is therefore
$\cV_2+\Span\{\beta_S+\sym_2(a_S):\ \emptyset\not=S\subseteq\{1,\dots,t_1\}\}$,
and it remains to identify the second summand.  Taking $S=\{i\}$ gives
$\beta_{\{i\}}+\sym_2(a_{\{i\}})=b_i$.  Taking $S=\{i,k\}$ with $i<k$ gives
$b_i+b_k+a_ia_k$, and subtracting the two functions $b_i$ and $b_k$ just
obtained leaves $a_ia_k$.  Conversely, for an arbitrary nonempty $S$, we have
$\beta_S+\sym_2(a_S)=\sum_{i\in S}b_i+\sum_{\{i,k\}\subseteq S}a_ia_k$, which
lies in the span of the functions $b_i$ and $a_ia_k$.  The second summand is
therefore $\Span\{b_i,\ a_ia_k\}$.

Adding the two contributions, the punctured span on $B_{\zero,\zero}$ is
spanned by
\begin{equation}\label{eq:B00-basis}
\epsilon,\quad a_i\ (1\leq i\leq t_1),\quad d_j\ (1\leq j\leq t_2),\quad
b_i\ (1\leq i\leq t_1),\quad a_ia_k\ (1\leq i<k\leq t_1).
\end{equation}
These functions are pairwise distinct square-free monomials in the independent
binary variables of \eqref{eq:message-digits}, hence linearly independent, as
observed in Subsection \ref{subsec:digits-rank}.  Their number is
$$
1+t_1+t_2+t_1+\binom{t_1}{2}=t_2+\binom{t_1+2}{2},
$$
which is \eqref{eq:rho00}.
\end{proof}

The blocks with $u=\zero$ and $v\not=\zero$ have the same local rank, but
reach it in a different way.  They meet no coordinate of the $\Z_2$ part, so
$\cV_2$ is not available, and $\epsilon$ and the $b_i$ of
\eqref{eq:B00-basis} are replaced by corrected versions of themselves, the
correction depending on $v$.

\begin{lemma}\label{lem:local-rank-0v}
Let $H^{t_1,t_2,1}$ be nonlinear and let $v\in\Z_2^{t_2}$ with
$v\not=\zero$.  Then,
\begin{equation}\label{eq:rho0v}
\rho_{H^{t_1,t_2,1}}(B_{\zero,v})=t_2+\binom{t_1+2}{2}.
\end{equation}
\end{lemma}

\begin{proof}
Put $R=\{j:v_j=1\}$, which is nonempty, and let $j_0=\min R$.  By item (ii)
of Lemma \ref{lem:block-count}, every binary coordinate of $B_{\zero,v}$ comes
from the $\Z_4$ part.  In the notation of Lemma \ref{lem:z4-digits}, it is
produced by an additive coordinate whose set of odd positions among the last
$t_2$ entries is exactly $R$, the set $S$ of odd positions among the first
$t_1$ entries being arbitrary.  By
\eqref{eq:z4-digits}, and abbreviating by $h_0$ the part of the top digit that
does not depend on $S$, $T$ and $R'$, the two digit functions are
\begin{equation}\label{eq:B0v-digits}
P=\alpha_S+\delta_R,\qquad
h=h_0+\beta_S+\alpha_T+\delta_{R'}+\sym_2(a_S)+\alpha_S\delta_R,
\qquad
h_0=\epsilon+\eta_R+\sym_2(d_R),
\end{equation}
and the normalization constraint of Lemma \ref{lem:z4-digits} reads: $\min
S\notin T$ when $S\not=\emptyset$, and $j_0\notin R'$ when $S=\emptyset$.  We
first exploit the coordinates with $S=\emptyset$ and then those with
$S\not=\emptyset$.

For the coordinates with $S=\emptyset$, formula \eqref{eq:B0v-digits} becomes
$$
P=\delta_R,\qquad h=h_0+\alpha_T+\delta_{R'},
$$
where $T$ is now an arbitrary subset of $\{1,\dots,t_1\}$ and $R'$ an
arbitrary subset of $\{1,\dots,t_2\}\setminus\{j_0\}$.  The free digits of the
column are the $\mu_i$ and the $\mu'_j$ with $j\not=j_0$, so the family of top
digits has the shape \eqref{eq:parameter-polynomial} with constant term $h_0$.
By Lemma \ref{lem:mobius}, these top digits contribute $h_0$, every $a_i$ and
every $d_j$ with $j\not=j_0$.  The bottom digit supplies the one function
still missing: since $j_0\in R$, we have
$P=\delta_R=d_{j_0}+\sum_{j\in R,\,j\not=j_0}d_j$, and all the summands other
than $d_{j_0}$ have just been obtained, so $d_{j_0}$ lies in the span as
well.

We turn to the coordinates with $S\not=\emptyset$.  Take such a coordinate and
add $h_0$ to its top digit. By \eqref{eq:B0v-digits}, this leaves
$\beta_S+\alpha_T+\delta_{R'}+\sym_2(a_S)+\alpha_S\delta_R$.  Reducing modulo
the functions $a_i$ and $d_j$ already obtained removes $\alpha_T$ and
$\delta_{R'}$, and we are left with
$$
\beta_S+\sym_2(a_S)+\alpha_S\delta_R
=\sum_{i\in S}\bigl(b_i+a_i\delta_R\bigr)+\sym_2(a_S)
=\sum_{i\in S}g_i+\sym_2(a_S),
$$
where $g_i=b_i+a_i\delta_R$.

Exactly as in the proof of Lemma \ref{lem:local-rank-00}, taking $S=\{i\}$
gives $g_i$, taking $S=\{i,k\}$ with $i<k$ gives $a_ia_k$, and every function
$\sum_{i\in S}g_i+\sym_2(a_S)$ lies in $\Span\{g_i,\ a_ia_k\}$.  The bottom
digit contributes nothing new, both of its summands being already
available.

Collecting what the two kinds of coordinate contribute, the punctured span on
$B_{\zero,v}$ is spanned by
\begin{equation}\label{eq:B0v-basis}
a_i\ (1\leq i\leq t_1),\quad d_j\ (1\leq j\leq t_2),\quad h_0,\quad
g_i\ (1\leq i\leq t_1),\quad a_ia_k\ (1\leq i<k\leq t_1),
\end{equation}
and these functions are linearly independent.  Indeed, suppose that a
$\Z_2$-linear combination of them vanishes.  The variable $\epsilon$ occurs in
$h_0=\epsilon+\eta_R+\sym_2(d_R)$ and in none of the other listed functions,
so the coefficient of $h_0$ must be zero.  Next, for each $i$ the variable
$b_i$ occurs in $g_i=b_i+a_i\delta_R$ and in no other listed function, so all
the coefficients of the $g_i$ must be zero as well.  What remains is a
vanishing combination of the functions $a_i$, $d_j$ and $a_ia_k$, which are
pairwise distinct square-free monomials, so all its coefficients vanish too.
The number of the listed functions is therefore the dimension we are after,
namely
$$
t_1+t_2+1+t_1+\binom{t_1}{2}=t_2+\binom{t_1+2}{2},
$$
by the same arithmetic as in the proof of Lemma \ref{lem:local-rank-00}.  This
is \eqref{eq:rho0v}.
\end{proof}

\begin{lemma}\label{lem:local-rank-u}
Let $H^{t_1,t_2,1}$ be nonlinear and let $u\in\Z_2^{t_1}$ with $u\not=\zero$.
Then, for every $v\in\Z_2^{t_2}$,
\begin{equation}\label{eq:rho1}
\rho_{H^{t_1,t_2,1}}(B_{u,v})=t_2+2+\binom{t_1+1}{2}.
\end{equation}
\end{lemma}

\begin{proof}
By item (i) of Lemma \ref{lem:block-count}, every binary coordinate of
$B_{u,v}$ comes from the $\Z_8$ part.  Put $S=\{i:u_i=1\}$, which is nonempty
because $u\not=\zero$, let $i_0=\min S$, and put $R=\{j:v_j=1\}$.  By
\eqref{eq:z8-label}, fixing the label $(u,v)$ fixes the parities of all the
entries of the column, so, in the notation \eqref{eq:column-digits}, the sets
$S$ and $R$ are the same for all the additive coordinates of the block.
Normalization forces $q_{i_0}=1$, that is,
$\mu_{i_0}=\nu_{i_0}=0$, whereas the digits $\mu_i$ and $\nu_i$ with
$i\not=i_0$ and all the digits $\mu'_j$ remain free; hence $T$ and $U$ range
over all the subsets of $\{1,\dots,t_1\}\setminus\{i_0\}$, and $R'$ over all
the subsets of $\{1,\dots,t_2\}$.  By Lemma \ref{lem:z8-digits}, the punctured
span is
generated by the three digit functions \eqref{eq:bottom-digit},
\eqref{eq:middle-digit} and \eqref{eq:top-digit} as $T$, $U$ and $R'$ vary.

We now determine the span from these three digits, applying Lemma
\ref{lem:mobius} with the free parameters $\mu_i$ and $\nu_i$ for
$i\not=i_0$, together with $\mu'_j$ for $1\leq j\leq t_2$.  The bottom digit
$P=\alpha_S$ does not involve those parameters at all, so it contributes the
single function $\alpha_S$.  The middle digit is
$Q=Q_0+\alpha_T=Q_0+\sum_{i\not=i_0}\mu_ia_i$, so it contributes $Q_0$
together with every $a_i$ with $i\not=i_0$.  For the top digit, we
substitute
$$
\alpha_U=\sum_{i\not=i_0}\nu_ia_i,\quad
\delta_{R'}=\sum_{j=1}^{t_2}\mu'_jd_j,\quad
\beta_T=\sum_{i\not=i_0}\mu_ib_i,\quad
\sym_2(a_T)=\sum_{\substack{i<k\\ i,k\not=i_0}}\mu_i\mu_k\,a_ia_k
$$
into \eqref{eq:top-digit} and obtain
$$
F=F_0+\sum_{i\not=i_0}\nu_ia_i+\sum_{j=1}^{t_2}\mu'_jd_j
   +\sum_{i\not=i_0}\mu_i\bigl(b_i+a_iQ_0\bigr)
   +\sum_{\substack{i<k\\ i,k\not=i_0}}\mu_i\mu_k\,a_ia_k,
$$
where the fourth sum collects the two terms $\beta_T$ and $\alpha_TQ_0$ of
\eqref{eq:top-digit}.  The top digits therefore contribute $F_0$, the
functions $a_i$ with $i\not=i_0$, all the $d_j$, the functions
$\Psi_i=b_i+a_iQ_0$ with $i\not=i_0$, and the products $a_ia_k$ with $i<k$ and
$i,k\not=i_0$.  Finally, the missing function $a_{i_0}$ is supplied by the
bottom digit: since $i_0\in S$,
$$
\alpha_S=a_{i_0}+\sum_{i\in S,\ i\not=i_0}a_i,
$$
and every summand other than $a_{i_0}$ has already been obtained.

Summarising, the punctured span on $B_{u,v}$ is spanned by
\begin{equation}\label{eq:Buv-basis}
a_i\ (1\leq i\leq t_1),\quad d_j\ (1\leq j\leq t_2),\quad Q_0,\quad F_0,
\quad \Psi_i\ (i\not=i_0),\quad a_ia_k\ (i,k\not=i_0,\ i<k),
\end{equation}
and we check that these $t_1+t_2+2+(t_1-1)+\binom{t_1-1}{2}$ functions are
linearly independent.  Suppose that a $\Z_2$-linear combination of them
vanishes.  Only $F_0$ involves the variables $c_i$, and it does so through the
term $\gamma_S=\sum_{i\in S}c_i$ of \eqref{eq:Q0F0}, which is a nonzero
function because $S\not=\emptyset$; hence the coefficient of $F_0$ must be
zero.  Among the remaining functions, collect now the terms that are one of
the monomials $b_1,\dots,b_{t_1}$: there are none for $a_i$, $d_j$ and
$a_ia_k$, they add up to $\beta_S$ for $Q_0$, and they reduce to $b_i$ for
$\Psi_i=b_i+a_iQ_0$ with $i\not=i_0$, since all the other terms of $\Psi_i$
have degree at least two.  Since $i_0\in S$, the function
$\beta_S=\sum_{i\in S}b_i$ involves $b_{i_0}$, which appears in no $\Psi_i$.
Therefore, the $t_1$ functions $\beta_S$ and $b_i$ with $i\not=i_0$ are
linearly independent, and the coefficients of $Q_0$ and of the $\Psi_i$ must
vanish as well.  What remains is a vanishing combination of the functions
$a_i$, $d_j$ and $a_ia_k$, which are pairwise distinct square-free monomials,
so all the coefficients are zero.  The punctured span therefore has dimension
\begin{equation*}
\begin{aligned}
t_1+t_2+2+(t_1-1)+\binom{t_1-1}{2}
&=t_2+\frac{4t_1+2+(t_1-1)(t_1-2)}{2}\\
&=t_2+\frac{t_1^2+t_1+4}{2}
=t_2+2+\binom{t_1+1}{2},
\end{aligned}
\end{equation*}
which is \eqref{eq:rho1}.
\end{proof}

\begin{example}\label{ex:local-H201}
For $H^{2,0,1}$, we have $t=6$ by \eqref{eq:length-relation}, so the length is
$2^6=64$, and Lemma \ref{lem:block-count} gives $4$ blocks of size $16$.
There are no rows of order $4$ here, so a label is just a vector
$u\in\Z_2^2$, the message is $\epsilon\ww_0+x_1\ww_1+x_2\ww_2$ with
$x_i=a_i+2b_i+4c_i$, and all the functions $\delta_R$ and $\eta_R$ vanish.

For the block $B_{\zero}$, formula \eqref{eq:B00-basis} gives the six
functions $\epsilon$, $a_1$, $a_2$, $b_1$, $b_2$ and $a_1a_2$, so
$\rho_0=6$, in agreement with $t_2+\binom{t_1+2}{2}=\binom{4}{2}=6$.

For a block $B_u$ with $u\not=\zero$, formula \eqref{eq:Buv-basis} gives $a_1$
and $a_2$, then $Q_0$ and $F_0$, then the single function $\Psi_i$ with
$i\not=i_0$. There is no product $a_ia_k$ left, since only one index differs
from $i_0$.  This is $5=t_2+2+\binom{t_1+1}{2}=2+\binom{3}{2}$ functions, as
predicted.  Writing them out from \eqref{eq:Q0F0} for the three admissible
labels, we obtain
$$
\begin{aligned}
u&=(1,0):&\ S&=\{1\},\ i_0=1,&\ &Q_0=b_1,\ F_0=\epsilon+c_1,\
 \Psi_2=b_2+a_2b_1,\\
u&=(0,1):&\ S&=\{2\},\ i_0=2,&\ &Q_0=b_2,\ F_0=\epsilon+c_2,\
 \Psi_1=b_1+a_1b_2,\\
u&=(1,1):&\ S&=\{1,2\},\ i_0=1,&\ &Q_0=b_1+b_2+a_1a_2,\\
& & & & &F_0=\epsilon+c_1+c_2+b_1b_2+(b_1+b_2)a_1a_2,
\end{aligned}
$$
and in the last case $\Psi_2=b_2+a_2Q_0=b_2+a_2b_1+a_2b_2+a_1a_2$, where we
used $a_2\cdot a_1a_2=a_1a_2$.  Hence,
$\cR(H^{2,0,1})=\{\!\{6^{[1]},5^{[3]}\}\!\}$.
\end{example}

%%%%%%%%%%%%%%%%%%%%%%%%%%%%%%%%%%%%%%%%%%%%%%%%%%%%%%%%%%%%%%%%%%%%%%%%%%%%%
\subsection{The profile of $H^{t_1,t_2,t_3}$}\label{subsec:profile-theorem}
%%%%%%%%%%%%%%%%%%%%%%%%%%%%%%%%%%%%%%%%%%%%%%%%%%%%%%%%%%%%%%%%%%%%%%%%%%%%%

It remains to add the rows of order $2$, that is, to pass from $t_3=1$ to a
general $t_3$.  The next lemma describes the effect of the duplication
\eqref{eq:construction-order2} on the whole block structure.  It is the exact
analogue, for the profile, of the fact that a duplication raises the rank by
one, and we state it for an arbitrary binary code, since nothing in it is
special to the present family.

\begin{lemma}\label{lem:plotkin-blocks}
Let $C\subseteq\Z_2^N$ be a binary code with $\zero\in C$ and $\one\in\K(C)$,
and let
\begin{equation}\label{eq:plotkin-extension}
C^{+}=\{(\xx,\xx):\xx\in C\}\cup\{(\xx,\xx+\one):\xx\in C\}
\end{equation}
be its Plotkin extension.  Then,
\begin{equation}\label{eq:plotkin-kernel-span}
\begin{aligned}
\K(C^{+})&=\{(\xx,\xx)+\varepsilon(\zero,\one):\xx\in\K(C),\ \varepsilon\in\Z_2\},\\
\langle C^{+}\rangle&=\{(\zz,\zz)+\varepsilon(\zero,\one):
\zz\in\langle C\rangle,\ \varepsilon\in\Z_2\}.
\end{aligned}   
\end{equation}
Every kernel block $B$ of $C$ gives rise to exactly two kernel blocks $B^{0}$
and $B^{1}$ of $C^{+}$, one inside each copy, and
\begin{equation}\label{eq:plotkin-local-rank}
|B^{0}|=|B^{1}|=|B|,\qquad
\rho_{C^{+}}(B^{0})=\rho_{C^{+}}(B^{1})=\rho_{C}(B).
\end{equation}
In particular, $\cR(C^{+})$ is obtained from $\cR(C)$ by doubling all the
multiplicities, and $\one\in\K(C^{+})$.
\end{lemma}

\begin{proof}
Write $\Delta=\{(\xx,\xx):\xx\in C\}$ for the diagonal copy of $C$; then,
\eqref{eq:plotkin-extension} reads
$C^{+}=\Delta\cup\bigl(\Delta+(\zero,\one)\bigr)$.  We determine in turn the
span, the kernel, and then the blocks together with the local ranks.

The span of $\Delta$ is $\{(\zz,\zz):\zz\in\langle C\rangle\}$.  Moreover,
$(\zero,\one)$ is the sum of the two elements $(\xx,\xx)$ and $(\xx,\xx+\one)$
of $C^{+}$, hence it lies in $\langle C^{+}\rangle$.  Therefore,
$\langle C^{+}\rangle$ contains the right-hand side of the second identity in
\eqref{eq:plotkin-kernel-span}; and it is contained in it, because that
right-hand side is a linear space containing $C^{+}$.  This proves the second
identity.

For the kernel, observe first that every element of $C^{+}$ has its first half
in $C$, by \eqref{eq:plotkin-extension}.  Let $(\zz_1,\zz_2)\in\K(C^{+})$.
Since $\zero\in C^{+}$, the vector $(\zz_1,\zz_2)$ itself lies in $C^{+}$, so
$\zz_1\in C$.  Adding $(\zz_1,\zz_2)$ to $(\xx,\xx)$ for an arbitrary
$\xx\in C$ gives an element of $C^{+}$ whose first half is $\zz_1+\xx$, so
$\zz_1+C\subseteq C$ and, $C$ being finite, $\zz_1\in\K(C)$.  This proves that
$\K(C^{+})$ is
contained in the right-hand side of the first identity in
\eqref{eq:plotkin-kernel-span}.

Conversely, let $\xx\in\K(C)$ and let $\yy\in C$ be arbitrary.  Adding
$(\xx,\xx)$ to $(\yy,\yy)$ gives $(\xx+\yy,\xx+\yy)$, and adding it to
$(\yy,\yy+\one)$ gives $(\xx+\yy,\xx+\yy+\one)$.  In both cases
$\xx+\yy\in C$, so the result lies in $C^{+}$, and, $C^{+}$ being finite,
$(\xx,\xx)\in\K(C^{+})$.  Similarly, adding $(\zero,\one)$ to $(\yy,\yy)$
gives $(\yy,\yy+\one)\in C^{+}$, and adding it to $(\yy,\yy+\one)$ gives
$(\yy,\yy)\in C^{+}$; hence $(\zero,\one)\in\K(C^{+})$ as well.  This proves
the first identity in \eqref{eq:plotkin-kernel-span}.  In particular, taking
$\xx=\one$, which lies in $\K(C)$ by hypothesis, shows that
$\one=(\one,\one)\in\K(C^{+})$.

We can now describe the blocks.  The kernel word $(\zero,\one)$ takes the
value $0$ at
every coordinate of the first copy and the value $1$ at every coordinate of
the second one, so no coordinate of the first copy can be equivalent to a
coordinate of the second one. Within one copy, a kernel word
$(\xx,\xx)+\varepsilon(\zero,\one)$ takes at
the position $i$ the value $x_i$ or $x_i+\varepsilon$, according to the copy.
The contribution of $\varepsilon$ is the same at all the positions of that
copy, so two of them are equivalent for $C^{+}$ exactly when they are
equivalent for $C$.
Therefore, every block $B$ of $C$ gives exactly two blocks
of $C^{+}$, namely its copy $B^{0}$ in the first half and its copy $B^{1}$ in
the second half, and these have the same size as $B$.

It remains to compute the local ranks.  Restrict the description of
$\langle C^{+}\rangle$ given in \eqref{eq:plotkin-kernel-span} to $B^{0}$.
The diagonal part
$(\zz,\zz)$ restricts to $\zz|_B$, so it gives $\langle C\rangle|_B$, whereas
the extra vector $(\zero,\one)$ restricts to $\zero$ and adds nothing; hence
$\rho_{C^{+}}(B^{0})=\rho_C(B)$.  Restricting to $B^{1}$, the diagonal part
again gives $\langle C\rangle|_B$, and the extra vector now restricts to the
all-one vector of length $|B|$.  That vector is already present, because
$\one\in\K(C)\subseteq\langle C\rangle$.
Hence, $\rho_{C^{+}}(B^{1})=\rho_C(B)$ as well, which is
\eqref{eq:plotkin-local-rank}.
\end{proof}

\begin{theorem}\label{thm:profile}
Let $t_1\geq1$, $t_2\geq0$ and $t_3\geq1$, and assume that $H^{t_1,t_2,t_3}$
is nonlinear, that is, $(t_1,t_2)\not=(1,0)$.  Put
\begin{equation}\label{eq:rho-def}
\rho_0=t_2+\binom{t_1+2}{2},
\qquad
\rho_1=t_2+2+\binom{t_1+1}{2}.
\end{equation}
Then, the kernel partition of $H^{t_1,t_2,t_3}$ has exactly
$2^{t_1+t_2+t_3-1}$ blocks, all of size $2^{2t_1+t_2}$, and
\begin{equation}\label{eq:profile-formula}
\cR\bigl(H^{t_1,t_2,t_3}\bigr)
=\bigl\{\!\bigl\{\ \rho_0^{\,[2^{t_2+t_3-1}]},\
\rho_1^{\,[(2^{t_1}-1)2^{t_2+t_3-1}]}\ \bigr\}\!\bigr\}.
\end{equation}
Moreover, $\rho_0-\rho_1=t_1-1$, so the profile is constant if and only if
$t_1=1$, in which case
\begin{equation}\label{eq:profile-t1-one}
\cR\bigl(H^{1,t_2,t_3}\bigr)
=\bigl\{\!\bigl\{(t_2+3)^{\,[2^{t_2+t_3}]}\bigr\}\!\bigr\}.
\end{equation}
\end{theorem}

\begin{proof}
We treat first the case $t_3=1$.  By Lemma \ref{lem:block-count}, there are
$2^{t_1+t_2}$ blocks, all of size $2^{2t_1+t_2}$, one for each label
$(u,v)\in\Z_2^{t_1}\times\Z_2^{t_2}$.  Among those labels, exactly $2^{t_2}$
have $u=\zero$, namely one for each choice of $v$, and the remaining
$(2^{t_1}-1)2^{t_2}$ have $u\not=\zero$.  By Lemmas \ref{lem:local-rank-00}
and \ref{lem:local-rank-0v}, every block with $u=\zero$ has local rank
$\rho_0$, whichever value $v$ takes, and by Lemma \ref{lem:local-rank-u}, every
block with $u\not=\zero$ has local rank $\rho_1$.  This is exactly
\eqref{eq:profile-formula} for $t_3=1$, since $2^{t_2+t_3-1}=2^{t_2}$ there.

We now increase $t_3$, arguing by induction.  Put $C=H^{t_1,t_2,t_3}$ with
$t_3\geq1$ and assume the statement for $C$.  By
\eqref{eq:construction-order2}, an additive codeword of
$\cH^{t_1,t_2,t_3+1}$ is $(\uu\mid\uu)+\varepsilon\ww$, where
$\uu\in\cH^{t_1,t_2,t_3}$, $\varepsilon\in\Z_2$ and $\ww$ is the new row of
order $2$, which vanishes on the first copy of each of the three parts and
equals $\one$, $\two$ and $\four$ on the second copies.  Every entry of $\ww$
has the form $2^{s-1}\lambda$ inside its own part, so Lemma
\ref{lem:gray-shift} gives
$\Phi\bigl((\uu\mid\uu)+\varepsilon\ww\bigr)
=\Phi(\uu\mid\uu)+\varepsilon\Phi(\ww)$.  Moreover, $\phi_2(2)=(1,1)$ and
$\phi_3(4)=(1,1,1,1)$, so $\Phi(\ww)$ vanishes on the binary coordinates
coming from the first copies and equals $1$ on those coming from the second
ones.  Since \eqref{eq:construction-order2} interleaves the two copies part
by part, we let $\pi$ list the three second copies after the three first
ones, and then $\pi(H^{t_1,t_2,t_3+1})=C^{+}$.  The hypothesis
$\one\in\K(C)$ of Lemma \ref{lem:plotkin-blocks} holds for $t_3=1$ by
\eqref{eq:kernel-basis}, and for $t_3>1$ by the last assertion of that lemma
at the previous step.  Each step therefore doubles the number of blocks and
leaves their sizes and their local ranks unchanged.  Passing from $t_3=1$ to a
general $t_3$ requires $t_3-1$ steps, which multiplies both multiplicities of
\eqref{eq:profile-formula} by $2^{t_3-1}$ and leaves $\rho_0$, $\rho_1$ and
the block size unchanged; the profile so computed is the profile of
$H^{t_1,t_2,t_3}$ by Lemma \ref{lem:profile-invariant}.  The total number of
blocks is $2^{t_1+t_2}2^{t_3-1}=2^{t_1+t_2+t_3-1}$, which is
$2^{\kernel(H^{t_1,t_2,t_3})-1}$ by item (ii) of Theorem
\ref{thm:recalled-linearity-kernel}.

It remains to consider the difference between the two values. From \eqref{eq:rho-def},
$$
\rho_0-\rho_1=\binom{t_1+2}{2}-\binom{t_1+1}{2}-2
=\frac{(t_1+2)(t_1+1)-(t_1+1)t_1}{2}-2=(t_1+1)-2=t_1-1.
$$
Both multiplicities in \eqref{eq:profile-formula} are nonzero, since
$t_2+t_3-1\geq0$ and $2^{t_1}-1\geq1$, so both values are actually attained
and the profile is constant if and only if $\rho_0=\rho_1$.  Hence, the
profile is constant exactly when $t_1=1$.  In that case
$\rho_0=t_2+\binom{3}{2}=t_2+3$, and the two multiplicities in
\eqref{eq:profile-formula} add up to
$2^{t_2+t_3-1}+(2^1-1)2^{t_2+t_3-1}=2^{t_2+t_3}$, which gives
\eqref{eq:profile-t1-one}.
\end{proof}

\begin{example}\label{ex:profile-small}
For $H^{1,1,1}$, we have $\rho_0=\rho_1=1+3=4$, so the profile is constant, and
by Example \ref{ex:H111-blocks}, there are four blocks of size $8$; hence
$\cR(H^{1,1,1})=\{\!\{4^{[4]}\}\!\}$, the multiset recording that there are
four blocks even though they all carry the same local rank.  For $H^{2,0,1}$,
we computed in Example \ref{ex:local-H201} that $\rho_0=6$ and $\rho_1=5$,
with one block of local rank $6$ and three of local rank $5$, so
$\cR(H^{2,0,1})=\{\!\{6^{[1]},5^{[3]}\}\!\}$.  Observe how much more this says
than the pair $(r,k)=(12,3)$ given by Theorem \ref{thm:rank-main} and item
(ii) of Theorem \ref{thm:recalled-linearity-kernel}.  Being nonconstant, the
profile already distinguishes $H^{2,0,1}$ from every member of the family with
$t_1=1$.  Passing to $t_3=2$, Lemma \ref{lem:plotkin-blocks} doubles the
multiplicities without changing the values, so
$\cR(H^{2,0,2})=\{\!\{6^{[2]},5^{[6]}\}\!\}$, in accordance with
\eqref{eq:profile-formula}.
\end{example}

%%%%%%%%%%%%%%%%%%%%%%%%%%%%%%%%%%%%%%%%%%%%%%%%%%%%%%%%%%%%%%%%%%%%%%%%%%%%%
\section{The complete classification inside the family}
\label{sec:internal-classification}
%%%%%%%%%%%%%%%%%%%%%%%%%%%%%%%%%%%%%%%%%%%%%%%%%%%%%%%%%%%%%%%%%%%%%%%%%%%%%

We can now classify the family $H^{t_1,t_2,t_3}$ completely.  Let us first see
why the two invariants already available are not enough.  By item (ii) of
Theorem \ref{thm:recalled-linearity-kernel}, the dimension of the kernel of a
nonlinear member is $t_1+t_2+t_3$, so two codes of the same length whose types
have different sums are nonequivalent; but the sum does not determine the
type, and, by Proposition \ref{prop:all-collisions}, the rank does not
separate the remaining cases either.  In \cite{Z2Z4Z8Linearity}, the
classification was therefore completed only for $3\leq t\leq11$, and the pairs
not separated by either invariant were distinguished using {\sc Magma} \cite{Magma}.

\begin{example}\label{ex:no-kernel-no-rank}
By Theorem \ref{thm:recalled-linearity-kernel}, the admissible types of length
$2^9$ are the triples $(t_1,t_2,t_3)$ with $t_1\geq1$, $t_2\geq0$, $t_3\geq1$
and $3t_1+2t_2+t_3=10$, and the only linear one is $H^{1,0,7}$.  The nonlinear
$\Z_2\Z_4\Z_8$-linear Hadamard codes of length $2^9$ are therefore
$$H^{1,1,5},\, H^{1,2,3},\, H^{1,3,1},\, H^{2,0,4},\, H^{2,1,2},\, H^{3,0,1}.$$  Their kernel dimensions $t_1+t_2+t_3$ are $7$, $6$, $5$, $6$,
$5$ and $4$, and by Theorem \ref{thm:rank-main} their ranks are $12$, $15$,
$19$, $15$, $20$ and $26$.  Neither invariant alone classifies these six
codes, since each of the two lists has repetitions; and neither does the pair
$(r,k)$, because $H^{1,2,3}$ and $H^{2,0,4}$ both have $(r,k)=(15,6)$.  In
\cite{Z2Z4Z8Linearity}, those two had to be separated by a computer
equivalence test.
\end{example}

The kernel-block rank profile removes this obstruction at once: we show that
it recovers $t_1$, and hence, together with the length and the dimension of
the kernel, the whole type.

%%%%%%%%%%%%%%%%%%%%%%%%%%%%%%%%%%%%%%%%%%%%%%%%%%%%%%%%%%%%%%%%%%%%%%%%%%%%%
\subsection{Recovering the parameters}\label{subsec:recovering}
%%%%%%%%%%%%%%%%%%%%%%%%%%%%%%%%%%%%%%%%%%%%%%%%%%%%%%%%%%%%%%%%%%%%%%%%%%%%%

\begin{theorem}\label{thm:complete-internal-classification}
Let $t\geq3$ and let $(t_1,t_2,t_3)$ and $(t'_1,t'_2,t'_3)$ be two types with
$t_1,t'_1\geq1$, $t_2,t'_2\geq0$, $t_3,t'_3\geq1$ and
$3t_1+2t_2+t_3=3t'_1+2t'_2+t'_3=t+1$.  Then, $H^{t_1,t_2,t_3}$ and
$H^{t'_1,t'_2,t'_3}$ are equivalent if and only if
$(t_1,t_2,t_3)=(t'_1,t'_2,t'_3)$.  In other words, distinct admissible types
give nonequivalent codes, at every length.
\end{theorem}

\begin{proof}
The condition is sufficient, since equal types denote the same code.
We prove that it is necessary, so assume from now on that the two codes are
equivalent, and let us recover the three parameters one after the other.

We may assume that both codes are nonlinear.  Indeed, by item (i) of Theorem
\ref{thm:recalled-linearity-kernel}, exactly one member of the family of
length $2^t$ is linear, namely $H^{1,0,t-2}$.  Linearity is preserved by
equivalence, so if one of our two codes is linear then so is the other, and
both types equal $(1,0,t-2)$, which is the desired conclusion.

The dimension of the kernel comes next.  Since the two codes are equivalent,
their kernels have the same dimension, say $k$, by Lemma
\ref{lem:rank-kernel-invariance}.  By item (ii) of Theorem
\ref{thm:recalled-linearity-kernel} applied to each of them,
\begin{equation}\label{eq:k-sum}
k=t_1+t_2+t_3=t'_1+t'_2+t'_3.
\end{equation}

We recover $t_1$ from the profile.  By Lemma
\ref{lem:profile-invariant}, equivalent codes have the same kernel-block rank
profile, so $\cR(H^{t_1,t_2,t_3})=\cR(H^{t'_1,t'_2,t'_3})$.  Suppose first
that these profiles are constant.  By the last part of Theorem
\ref{thm:profile},
a profile is constant exactly when the first parameter equals $1$, so
$t_1=1=t'_1$.  Suppose now that they are not constant.  Then, again by Theorem
\ref{thm:profile}, the profile of the first code takes exactly the two values
$\rho_0$ and $\rho_1$ of \eqref{eq:rho-def}, with $\rho_0-\rho_1=t_1-1>0$,
and that of the second one takes exactly the two values $\rho'_0$ and
$\rho'_1$ obtained from \eqref{eq:rho-def} with $(t'_1,t'_2)$ in place of
$(t_1,t_2)$, with $\rho'_0-\rho'_1=t'_1-1>0$.  The two multisets coincide, so
they have the same largest and the same smallest element, that is,
$\rho_0=\rho'_0$ and $\rho_1=\rho'_1$; hence
$$
t_1-1=\rho_0-\rho_1=\rho'_0-\rho'_1=t'_1-1.
$$
In both cases $t_1=t'_1$.

Finally, we recover $t_2$ and $t_3$.  Subtracting \eqref{eq:k-sum} from
the length relation \eqref{eq:length-relation}, that is, from
$t+1=3t_1+2t_2+t_3$, we obtain $t+1-k=2t_1+t_2$, and therefore
\begin{equation*}
t_2=t+1-k-2t_1,\qquad t_3=k-t_1-t_2.
\end{equation*}
Since $t$, $k$ and $t_1$ coincide for the two codes, so do $t_2$ and $t_3$.
\end{proof}

For a nonlinear member with $t_1\geq2$, the profile alone already determines
the type.  Indeed, the two values of the profile give $t_1$ as in Step 3
above, the multiplicity of $\rho_0$ is $2^{k-t_1-1}$ by
\eqref{eq:profile-formula} and \eqref{eq:k-sum}, which recovers $k$, and then
$\rho_0=t_2+\binom{t_1+2}{2}$ recovers $t_2$ and \eqref{eq:k-sum} recovers
$t_3$.

%%%%%%%%%%%%%%%%%%%%%%%%%%%%%%%%%%%%%%%%%%%%%%%%%%%%%%%%%%%%%%%%%%%%%%%%%%%%%
\subsection{The number of classes and the two collisions}\label{subsec:counting}
%%%%%%%%%%%%%%%%%%%%%%%%%%%%%%%%%%%%%%%%%%%%%%%%%%%%%%%%%%%%%%%%%%%%%%%%%%%%%

Counting the equivalence classes is now the same thing as counting the
admissible types.

\begin{corollary}\label{cor:number-classes}
For $t\geq3$, the number of pairwise nonequivalent $\Z_2\Z_4\Z_8$-linear
Hadamard codes $H^{t_1,t_2,t_3}$ of length $2^t$ is
\begin{equation}\label{eq:number-classes}
\cA_t=\sum_{j=1}^{\lfloor t/3\rfloor}
\left(\left\lfloor\frac{t-3j}{2}\right\rfloor+1\right)
=\left\lfloor\frac{t^2+6}{12}\right\rfloor.
\end{equation}
\end{corollary}

\begin{proof}
By Theorem \ref{thm:complete-internal-classification}, two codes of the family
of length $2^t$ are equivalent if and only if their types coincide, so the
equivalence classes are in bijection with the admissible types, whose number
is $\cA_t$; the two expressions in \eqref{eq:number-classes} are proved to be
equal in \cite[Proposition 4.12]{Z2Z4Z8Rank}.
\end{proof}

For $3\leq t\leq15$, formula \eqref{eq:number-classes} gives
$$
\cA_t=1,\ 1,\ 2,\ 3,\ 4,\ 5,\ 7,\ 8,\ 10,\ 12,\ 14,\ 16,\ 19,
$$
in agreement with the numbers of nonequivalent codes obtained in
\cite{Z2Z4Z8Linearity} with the help of {\sc Magma} for $3\leq t\leq11$.  The
point of Corollary \ref{cor:number-classes} is that the same count is now
proved for every $t$, and that it counts equivalence classes and not merely
admissible types.  By contrast, the number of distinct pairs $(r,k)$ is
$\cA_t-1$ for $7\leq t\leq9$ and $\cA_t-2$ for $t\geq10$, by
\cite[Corollary 4.13]{Z2Z4Z8Rank}, and that difference is exactly what the
profile repairs.

Theorem \ref{thm:complete-internal-classification} makes no reference to the
rank, so it is worth recording separately what the profile does on the two
families of pairs that the rank and the dimension of the kernel cannot
separate.

\begin{corollary}\label{cor:collision-profiles}
For every $t\geq7$ in the first line below, and for every $t\geq10$ in the
second one,
$$
\begin{aligned}
\cR\bigl(H^{1,2,t-6}\bigr)&=\{\!\{5^{[2^{t-4}]}\}\!\},&
\cR\bigl(H^{2,0,t-5}\bigr)&=\{\!\{6^{[2^{t-6}]},5^{[3\cdot2^{t-6}]}\}\!\},\\
\cR\bigl(H^{2,2,t-9}\bigr)&=\{\!\{8^{[2^{t-8}]},7^{[3\cdot2^{t-8}]}\}\!\},&
\cR\bigl(H^{3,0,t-8}\bigr)&=\{\!\{10^{[2^{t-9}]},8^{[7\cdot2^{t-9}]}\}\!\}.
\end{aligned}
$$
In particular, the two members of each of the two collision families of
Proposition \ref{prop:all-collisions} have different profiles, and hence they
are nonequivalent, for every $t\geq7$ and every $t\geq10$, respectively.
\end{corollary}

\begin{proof}
We apply Theorem \ref{thm:profile} to the four types in turn, using
\eqref{eq:rho-def} for the two values and \eqref{eq:profile-formula} for the
two multiplicities.  For $(1,2,t-6)$, the profile is constant by
\eqref{eq:profile-t1-one}, with value $5$ and multiplicity $2^{t-4}$.  For
$(2,0,t-5)$, we get $\rho_0=6$ and $\rho_1=5$, with multiplicities $2^{t-6}$
and $3\cdot2^{t-6}$.  For $(2,2,t-9)$, we get $\rho_0=8$ and $\rho_1=7$, with
multiplicities $2^{t-8}$ and $3\cdot2^{t-8}$.  For $(3,0,t-8)$, we get
$\rho_0=10$ and $\rho_1=8$, with multiplicities $2^{t-9}$ and
$7\cdot2^{t-9}$.

In the first pair, one profile is constant and the other is not, and in the
second pair, the two profiles take different values; thus, in both cases the two
multisets differ, so Lemma \ref{lem:profile-invariant} shows that the two
codes are nonequivalent.
\end{proof}

It is instructive to see why the profile succeeds where the rank fails.  Both
invariants are built from the linear span $\langle C\rangle$, but the rank
reports only its total dimension, whereas the profile reports how that
dimension is distributed over the kernel blocks.  In the pair
$\bigl(H^{1,2,t-6},H^{2,0,t-5}\bigr)$ the two spans have the same dimension
$t+6$ and the two codes have the same number $2^{t-4}$ of kernel blocks, of
the same size $2^{4}$; but the blocks of $H^{1,2,t-6}$ all carry local rank
$5$, while those of $H^{2,0,t-5}$ carry local rank $6$ or $5$ according to the
label.  The same total dimension is distributed in two genuinely different
ways, and the multiset \eqref{eq:profile} is precisely what records the
difference.

Table \ref{tab:invariants} lists the three invariants for all the members of
the family of length $2^t$ with $5\leq t\leq12$.  It is the analogue, for the
present classification, of Tables 2--4 of \cite{Z2Z4Z8Linearity}, with two
differences: the ranks come from Theorem \ref{thm:rank-main} instead of from a
computer computation, and the profile column makes every nonequivalence
visible without any further test.  The codes marked with $\bowtie$ are those
belonging to a pair predicted by Proposition \ref{prop:all-collisions}, and
each such pair had to be separated with {\sc Magma} in
\cite{Z2Z4Z8Linearity}.  At
$t=9$, for instance, the profiles of $H^{1,2,3}$ and $H^{2,0,4}$ are
respectively constant and nonconstant, which answers the question left open
in Example \ref{ex:no-kernel-no-rank}; together with Corollary
\ref{cor:number-classes}, this confirms that there are $\cA_9=7$ classes at
that length.

\begin{table}[ht]
\scriptsize
\centering
\caption{Rank $r$, dimension of the kernel $k$ and kernel-block rank profile
$\cR$ of the $\Z_2\Z_4\Z_8$-linear Hadamard codes $H^{t_1,t_2,t_3}$ of length
$2^t$, for $5\leq t\leq9$ (left) and $10\leq t\leq12$ (right).  The unique
linear code of each length is marked with an asterisk, and the codes belonging
to a pair with equal $(r,k)$ are marked with $\bowtie$.  Items (ii) and (iii)
of Theorem
\ref{thm:recalled-linearity-kernel} do not apply to the linear code
$H^{1,0,t-2}$, for which $r=k=t+1$.}
\label{tab:invariants}
\begin{minipage}[t]{0.46\textwidth}
\vspace{0pt}
\centering
\begin{tabular}{@{}clccl@{}}
\toprule
$t$ & Code & $r$ & $k$ & $\cR$\\
\midrule
\multirow{2}{*}{$5$}
 & $H^{1,0,3}{}^{*}$ & $6$ & $6$ & $\{\!\{1^{[32]}\}\!\}$\\
 & $H^{1,1,1}$ & $8$ & $3$ & $\{\!\{4^{[4]}\}\!\}$\\
\midrule
\multirow{3}{*}{$6$}
 & $H^{1,0,4}{}^{*}$ & $7$ & $7$ & $\{\!\{1^{[64]}\}\!\}$\\
 & $H^{1,1,2}$ & $9$ & $4$ & $\{\!\{4^{[8]}\}\!\}$\\
 & $H^{2,0,1}$ & $12$ & $3$ & $\{\!\{6^{[1]},5^{[3]}\}\!\}$\\
\midrule
\multirow{4}{*}{$7$}
 & $H^{1,0,5}{}^{*}$ & $8$ & $8$ & $\{\!\{1^{[128]}\}\!\}$\\
 & $H^{1,1,3}$ & $10$ & $5$ & $\{\!\{4^{[16]}\}\!\}$\\
 & $H^{1,2,1}{}^{\bowtie}$ & $13$ & $4$ & $\{\!\{5^{[8]}\}\!\}$\\
 & $H^{2,0,2}{}^{\bowtie}$ & $13$ & $4$ & $\{\!\{6^{[2]},5^{[6]}\}\!\}$\\
\midrule
\multirow{5}{*}{$8$}
 & $H^{1,0,6}{}^{*}$ & $9$ & $9$ & $\{\!\{1^{[256]}\}\!\}$\\
 & $H^{1,1,4}$ & $11$ & $6$ & $\{\!\{4^{[32]}\}\!\}$\\
 & $H^{1,2,2}{}^{\bowtie}$ & $14$ & $5$ & $\{\!\{5^{[16]}\}\!\}$\\
 & $H^{2,0,3}{}^{\bowtie}$ & $14$ & $5$ & $\{\!\{6^{[4]},5^{[12]}\}\!\}$\\
 & $H^{2,1,1}$ & $19$ & $4$ & $\{\!\{7^{[2]},6^{[6]}\}\!\}$\\
\midrule
\multirow{7}{*}{$9$}
 & $H^{1,0,7}{}^{*}$ & $10$ & $10$ & $\{\!\{1^{[512]}\}\!\}$\\
 & $H^{1,1,5}$ & $12$ & $7$ & $\{\!\{4^{[64]}\}\!\}$\\
 & $H^{1,2,3}{}^{\bowtie}$ & $15$ & $6$ & $\{\!\{5^{[32]}\}\!\}$\\
 & $H^{1,3,1}$ & $19$ & $5$ & $\{\!\{6^{[16]}\}\!\}$\\
 & $H^{2,0,4}{}^{\bowtie}$ & $15$ & $6$ & $\{\!\{6^{[8]},5^{[24]}\}\!\}$\\
 & $H^{2,1,2}$ & $20$ & $5$ & $\{\!\{7^{[4]},6^{[12]}\}\!\}$\\
 & $H^{3,0,1}$ & $26$ & $4$ & $\{\!\{10^{[1]},8^{[7]}\}\!\}$\\
\bottomrule
\end{tabular}
\end{minipage}
\hfill
\begin{minipage}[t]{0.52\textwidth}
\vspace{0pt}
\centering
\begin{tabular}{@{}clccl@{}}
\toprule
$t$ & Code & $r$ & $k$ & $\cR$\\
\midrule
\multirow{8}{*}{$10$}
 & $H^{1,0,8}{}^{*}$ & $11$ & $11$ & $\{\!\{1^{[1024]}\}\!\}$\\
 & $H^{1,1,6}$ & $13$ & $8$ & $\{\!\{4^{[128]}\}\!\}$\\
 & $H^{1,2,4}{}^{\bowtie}$ & $16$ & $7$ & $\{\!\{5^{[64]}\}\!\}$\\
 & $H^{1,3,2}$ & $20$ & $6$ & $\{\!\{6^{[32]}\}\!\}$\\
 & $H^{2,0,5}{}^{\bowtie}$ & $16$ & $7$ & $\{\!\{6^{[16]},5^{[48]}\}\!\}$\\
 & $H^{2,1,3}$ & $21$ & $6$ & $\{\!\{7^{[8]},6^{[24]}\}\!\}$\\
 & $H^{2,2,1}{}^{\bowtie}$ & $27$ & $5$ & $\{\!\{8^{[4]},7^{[12]}\}\!\}$\\
 & $H^{3,0,2}{}^{\bowtie}$ & $27$ & $5$ & $\{\!\{10^{[2]},8^{[14]}\}\!\}$\\
\midrule
\multirow{10}{*}{$11$}
 & $H^{1,0,9}{}^{*}$ & $12$ & $12$ & $\{\!\{1^{[2048]}\}\!\}$\\
 & $H^{1,1,7}$ & $14$ & $9$ & $\{\!\{4^{[256]}\}\!\}$\\
 & $H^{1,2,5}{}^{\bowtie}$ & $17$ & $8$ & $\{\!\{5^{[128]}\}\!\}$\\
 & $H^{1,3,3}$ & $21$ & $7$ & $\{\!\{6^{[64]}\}\!\}$\\
 & $H^{1,4,1}$ & $26$ & $6$ & $\{\!\{7^{[32]}\}\!\}$\\
 & $H^{2,0,6}{}^{\bowtie}$ & $17$ & $8$ & $\{\!\{6^{[32]},5^{[96]}\}\!\}$\\
 & $H^{2,1,4}$ & $22$ & $7$ & $\{\!\{7^{[16]},6^{[48]}\}\!\}$\\
 & $H^{2,2,2}{}^{\bowtie}$ & $28$ & $6$ & $\{\!\{8^{[8]},7^{[24]}\}\!\}$\\
 & $H^{3,0,3}{}^{\bowtie}$ & $28$ & $6$ & $\{\!\{10^{[4]},8^{[28]}\}\!\}$\\
 & $H^{3,1,1}$ & $37$ & $5$ & $\{\!\{11^{[2]},9^{[14]}\}\!\}$\\
\midrule
\multirow{12}{*}{$12$}
 & $H^{1,0,10}{}^{*}$ & $13$ & $13$ & $\{\!\{1^{[4096]}\}\!\}$\\
 & $H^{1,1,8}$ & $15$ & $10$ & $\{\!\{4^{[512]}\}\!\}$\\
 & $H^{1,2,6}{}^{\bowtie}$ & $18$ & $9$ & $\{\!\{5^{[256]}\}\!\}$\\
 & $H^{1,3,4}$ & $22$ & $8$ & $\{\!\{6^{[128]}\}\!\}$\\
 & $H^{1,4,2}$ & $27$ & $7$ & $\{\!\{7^{[64]}\}\!\}$\\
 & $H^{2,0,7}{}^{\bowtie}$ & $18$ & $9$ & $\{\!\{6^{[64]},5^{[192]}\}\!\}$\\
 & $H^{2,1,5}$ & $23$ & $8$ & $\{\!\{7^{[32]},6^{[96]}\}\!\}$\\
 & $H^{2,2,3}{}^{\bowtie}$ & $29$ & $7$ & $\{\!\{8^{[16]},7^{[48]}\}\!\}$\\
 & $H^{2,3,1}$ & $36$ & $6$ & $\{\!\{9^{[8]},8^{[24]}\}\!\}$\\
 & $H^{3,0,4}{}^{\bowtie}$ & $29$ & $7$ & $\{\!\{10^{[8]},8^{[56]}\}\!\}$\\
 & $H^{3,1,2}$ & $38$ & $6$ & $\{\!\{11^{[4]},9^{[28]}\}\!\}$\\
 & $H^{4,0,1}$ & $49$ & $5$ & $\{\!\{15^{[1]},12^{[15]}\}\!\}$\\
\bottomrule
\end{tabular}
\end{minipage}
\end{table}

%%%%%%%%%%%%%%%%%%%%%%%%%%%%%%%%%%%%%%%%%%%%%%%%%%%%%%%%%%%%%%%%%%%%%%%%%%%%%
\section{Conclusions and further research}\label{sec:conclusions}
%%%%%%%%%%%%%%%%%%%%%%%%%%%%%%%%%%%%%%%%%%%%%%%%%%%%%%%%%%%%%%%%%%%%%%%%%%%%%

We have introduced the kernel-block rank profile, an equivalence invariant
that couples the kernel of a binary code with the linear span: the kernel
partitions the coordinate set into blocks, and the invariant records the
multiset of the dimensions of the span punctured on those blocks.  We have
then computed it, in Theorem \ref{thm:profile}, for the family of
$\Z_2\Z_4\Z_8$-linear Hadamard codes $H^{t_1,t_2,t_3}$ constructed recursively
in \cite{Z2Z4Z8Construction}, whose linearity and kernel were determined in
\cite{Z2Z4Z8Linearity} and whose rank was determined in \cite{Z2Z4Z8Rank}.
The kernel partition has $2^{t_1+t_2+t_3-1}$ blocks, all of the same size
$2^{2t_1+t_2}$, and the profile takes at most two values, whose difference is
exactly $t_1-1$.  It therefore recovers $t_1$; the length and the dimension of
the kernel then recover $t_2$ and $t_3$.  Hence, two codes of the family with
the same length are equivalent if and only if their types coincide, and the
number of equivalence classes of length $2^t$ is $\lfloor(t^2+6)/12\rfloor$
for every $t\geq3$.  In particular, the two infinite families of pairs that
share the length, the rank and the dimension of the kernel are separated
uniformly, and every equivalence test carried out with {\sc Magma} in
\cite{Z2Z4Z8Linearity} is now replaced by a proof valid at every length.

Two comments delimit the scope of these statements.  First, the triple
$(t_1,t_2,t_3)$ does not determine, up to equivalence, every
$\Z_2\Z_4\Z_8$-additive Hadamard code of that abstract type: a generator
matrix of type $(4,6,12;2,0,1)$ whose Gray image has rank $13$, rather than the 
value $12$ of $H^{2,0,1}$, is exhibited in \cite{Z2Z4Z8Construction}.  All the
statements of this paper therefore concern the recursively constructed family,
together with the constructions already proved in \cite{Z2Z4Z8Construction} to
give permutation equivalent codes, and classifying all the
$\Z_2\Z_4\Z_8$-additive Hadamard codes of a given abstract type remains open.

Second, the classification obtained here is internal. It compares the codes
$H^{t_1,t_2,t_3}$ with one another.  A classification statement is, however,
only complete once these codes have also been compared with the Hadamard codes
that were already known, namely the $\Z_4$-linear, the $\Z_2\Z_4$-linear and
the $\Z_{2^s}$-linear ones.  That comparison needs the kernel-block rank
profiles of the comparison families, which are not computed here.  Once they are
known, the nonconstancy established above already separates every code
$H^{t_1,t_2,t_3}$ with $t_1\geq2$ from all of them at once, and only the case
$t_1=1$, where the profile is constant, requires a further argument.

Several natural questions remain.  The first one concerns ranks on unions of
several blocks.  The profile is the first level of a natural hierarchy: for
every $\ell\geq1$, we may consider the multiset of the numbers
$\dim\langle C\rangle|_{B_1\cup\cdots\cup B_\ell}$, taken over all the
$\ell$-subsets of kernel blocks.  For $\ell=1$, this is the kernel-block rank
profile studied here.  Whether the resulting refinement is strictly finer for
$\ell=2$, and how much of the structure of an arbitrary additive code is
recovered by the whole hierarchy, would be interesting to know.

The second question concerns other alphabets.  It would be natural to
generalize the construction and the present classification to
$\Z_2\Z_4\cdots\Z_{2^s}$-linear Hadamard codes with all the $\alpha_i$
different from zero, or even to $\Z_p\Z_{p^2}\cdots\Z_{p^s}$-linear
generalized Hadamard codes with $p$ prime, in the spirit of
\cite{HadamardZps,ZpZp2Construction,ZpZp2Classification}.  Definition
\ref{def:kernel-block-profile} applies verbatim in that generality and Lemma
\ref{lem:profile-invariant} holds without any change, so the invariant is
available. What has to be redone is the computation of the labels and of the
local ranks.  In the same direction, the existence of $\Z_4\Z_8$-additive
Hadamard codes, that is, the case $\alpha_1=0$, $\alpha_2\not=0$,
$\alpha_3\not=0$, is still open, whereas the case $\alpha_1\not=0$,
$\alpha_2=0$, $\alpha_3\not=0$ cannot occur \cite{Z2Z4Z8Linearity}.
%%%%%%%%%%%%%%%%%%%%%%%%%%%%%%%%%%%%%%%%%%%%%%%%%%%%%%%%%%%%%%%%%%%%%%%%%%%%%

\bibliographystyle{elsarticle-num}
\bibliography{manuscript}

@book{Key,
  title={Designs and Their Codes},
  author={Assmus, Edvard F and Key, Jennifer D},
  @number={103},
  year={1994},
  publisher={Cambridge University Press}
}

@article{BGH83,
  title={Algebraic techniques for nonlinear codes},
  author={Bauer, Heiko and Ganter, Bernhard and Hergert, Ferdinand},
  journal={Combinatorica},
  volume={3},
  number={1},
  pages={21--33},
  year={1983},
  publisher={Springer}
}

@article{PRV06,
  title={On the additive ({$\mathbb{Z}_4$}-linear and non-{$\mathbb{Z}_4$}-linear) {H}adamard codes: rank and kernel},
  author={Phelps, Kevin T and Rif{\`a}, Joseph and Villanueva, Merc{\`e}},
  journal={IEEE Transactions on Information Theory},
  volume={52},
  number={1},
  pages={316--319},
  year={2006},
  publisher={IEEE Press Piscataway, NJ, USA}
}

@article{ccsg,
  title={{$\mathbb{Z}_2\mathbb{Z}_4$}-linear codes: generator matrices and duality},
  author={Borges, Joaquim and Fern{\'a}ndez-C{\'o}rdoba, Cristina and Pujol, Jaume and Rif{\`a}, Josep and Villanueva, Merc{\`e}},
  journal={Designs, Codes and Cryptography},
  volume={54},
  number={2},
  pages={167--179},
  year={2010},
  publisher={Springer}
}

@book{BookZ2Z4,
  author    = {Joaquim Borges and
               Cristina Fern{\'{a}}ndez{-}C{\'{o}}rdoba and
               Jaume Pujol and
               Josep Rif{\`{a}} and
               Merc{\`{e}} Villanueva},
  title     = {$\Z_2\Z_4$-Linear Codes},
  publisher = {Springer},
  year      = {2022},
  isbn      = {978-3-031-05440-2},
}

@article{Codes2k,
  title={Every {$\mathbb{Z}_{2^k}$}-code is a binary propelinear code},
  author={Borges, Joaquim and Fern{\'a}ndez-C{\'o}rdoba, Cristina and Rif{\`a}, Josep},
  journal={Electronic Notes in Discrete Mathematics},
  volume={10},
  pages={100--102},
  year={2001},
  publisher={Elsevier}
}

@article{Magma,
  title={Handbook of {M}agma functions},
  author={Bosma, Wieb and Cannon, John J and Fieker, C and Steel, A},
  journal={Edition},
  volume={2.25},
  year={2020},
  url={http://magma.maths.usyd.edu.au/magma/},
  publisher={Unknown Publisher}
}

@article{Carlet,
  title={{$\mathbb{Z}_{2^k}$}-linear codes},
  author={Carlet, Claude},
  journal={IEEE Transactions on Information Theory},
  volume={44},
  number={4},
  pages={1543--1547},
  year={1998}
}

@article{Sole,
  title={The {$\mathbb{Z}_4$}-linearity of {K}erdock, {P}reparata, {G}oethals, and related codes},
  author={Hammons, A Roger and Kumar, P Vijay and Calderbank, A Robert and Sloane, Neil JA and Sol{\'e}, Patrick},
  journal={IEEE Transactions on Information Theory},
  volume={40},
  number={2},
  pages={301--319},
  year={1994},
  publisher={IEEE}
}

@article{Kro:2001:Z4_Had_Perf,
  title={{$\mathbb{Z}_4$}-linear {H}adamard and extended perfect codes},
  author={Krotov, Denis S},
  journal={Electronic Notes in Discrete Mathematics},
  volume={6},
  pages={107--112},
  year={2001},
  publisher={Elsevier}
}

@article{KV2015,
  title={Classification of the {$\mathbb{Z}_2\mathbb{Z}_4$}-linear {H}adamard codes and their automorphism groups},
  author={Krotov, Denis S and Villanueva, Merc{\`e}},
  journal={IEEE Transactions on Information Theory},
  volume={61},
  number={2},
  pages={887--894},
  year={2015},
  publisher={IEEE}
}

@book{WMcwill,
  title={The Theory of Error-correcting Codes},
  author={MacWilliams, Florence Jessie and Sloane, Neil James Alexander},
  year={1977},
  publisher={Elsevier}
}

@article{TwoWeightSole,
  title={On two-weight {$\mathbb{Z}_{2^k}$}-codes},
  author={Shi, Minjia and Sepasdar, Zahra and Alahmadi, Adel and Sol{\'e}, Patrick},
  journal={Designs, Codes and Cryptography},
  volume={86},
  number={6},
  pages={1201--1209},
  year={2018},
  publisher={Springer}
}

@article{Blake,
  title={Codes over integer residue rings},
  author={Blake, Ian F},
  journal={Information and Control},
  volume={29},
  number={4},
  pages={295--300},
  year={1975},
  publisher={Elsevier}
}

@article{HadamardZps,
  title={On the Linearity and Classification of {$\mathbb{Z}_{p^s}$}-Linear Generalized {H}adamard Codes},
  author={Bhunia, Dipak Kumar and Fern{\'a}ndez-C{\'o}rdoba, Cristina and Villanueva, Merc{\`e}},
  journal={Designs, Codes and Cryptography},
  volume={90},
  number={},
  pages={1037--1058},
  year={2022},
  publisher={Springer}
}

@article{ZpsEquivalance,
  title={On the Equivalence of {$\mathbb{Z}_{p^s}$}-Linear Generalized {H}adamard Codes},
  author={Bhunia, Dipak Kumar and Fern{\'a}ndez-C{\'o}rdoba, Cristina and Vela, Carlos and  Villanueva, Merc{\`e}},
  journal={Designs, Codes and Cryptography},
  volume={92},
  number={},
  pages={999--1022},
  year={2024},
  publisher={Springer}
}

@article{ZpZp2Construction,
  title={On the Constructions of  {$\mathbb{Z}_p\mathbb{Z}_{p^2}$}-Linear Generalized {H}adamard Codes},
  author={Bhunia, Dipak Kumar and Fern{\'a}ndez-C{\'o}rdoba, Cristina and Villanueva, Merc{\`e}},
  journal={Finite Fields and Their Applications},
  volume={83},
  number={},
  pages={102093},
  year={2022},
  publisher={Elsevier}
}

@article{ZpZp2Classification,
  title={Linearity and Classification of {$\Z_p\Z_{p^2}$}-Linear Generalized {H}adamard Codes},
  author={Bhunia, Dipak Kumar and Fern{\'a}ndez-C{\'o}rdoba, Cristina and Villanueva, Merc{\`e}},
  journal={Finite Fields and Their Applications},
  volume={86},
  number={},
  pages={102140},
  year={2023},
  publisher={Elsevier}
}

@article{Z2Z4Z8Construction,
  title={On recursive constructions of {$\Z_2\Z_4\Z_8$}-linear {H}adamard codes},
  author={Bhunia, Dipak Kumar and Fern{\'a}ndez-C{\'o}rdoba, Cristina and Villanueva, Merc{\`e}},
  journal={Advances in Mathematics of Communications},
  volume={18},
  number={2},
  pages={455--479},
  year={2024},
  publisher={American Institute of Mathematical Sciences}
}

@article{Z2Z4Z8Linearity,
  title={Linearity and classification of {$\Z_2\Z_4\Z_8$}-linear {H}adamard codes},
  author={Bhunia, Dipak Kumar and Fern{\'a}ndez-C{\'o}rdoba, Cristina and  Villanueva, Merc{\`e}},
  journal={Designs, Codes and Cryptography},
  volume={93},
  number={},
  pages={4567--4594},
  year={2025},
  publisher={Springer}
}

@article{Z2Z4Z8Rank,
  title={Rank and classification of {$\Z_2\Z_4\Z_8$}-linear {H}adamard codes},
  author={Bhunia, Dipak Kumar},
  journal={preprint arXiv:2609.06244},
  volume={},
  number={},
  pages={},
  year={2026},
  publisher={}
}

@article{dougherty,
  title={Codes over {$\mathbb{Z}_{2^k}$}, {G}ray map and self-dual codes},
  author={Dougherty, Steven T and Fern{\'a}ndez-C{\'o}rdoba, Cristina},
  journal={Advances in Mathematics of Communications},
  volume={5},
  number={4},
  pages={571--588},
  year={2011},
  publisher={American Institute of Mathematical Sciences}
}

@article{EquivZ2s,
  title={Equivalences among {$\mathbb{Z}_{2^s}$}-linear {H}adamard codes},
  author={Fern{\'a}ndez-C{\'o}rdoba, Cristina and Vela, Carlos and Villanueva, Merc{\`e}},
  journal={Discrete Mathematics},
  volume={343},
  number={3},
  pages={111721},
  year={2020},
  publisher={Elsevier}
}

@article{fernandez2019mathbb,
  title={On {$\mathbb{Z}_{8}$}-linear {H}adamard codes: rank and classification},
  author={Fern{\'a}ndez-C{\'o}rdoba, Cristina and Vela, Carlos and Villanueva, Merc{\`e}},
  journal={IEEE Transactions on Information Theory},
  volume={66},
  number={2},
  pages={970--982},
  year={2019},
  publisher={IEEE}
}

@article{KernelZ2s,
  title={On {$\mathbb{Z}_{2^s}$}-linear {H}adamard codes: kernel and partial classification},
  author={Fern{\'a}ndez-C{\'o}rdoba, Cristina and Vela, Carlos and Villanueva, Merc{\`e}},
  journal={Designs, Codes and Cryptography},
  volume={87},
  number={2-3},
  pages={417--435},
  year={2019},
  publisher={Springer}
}

@article{Krotov:2007,
  title={On  {$\mathbb{Z}_{2^k}$}-dual binary codes},
  author={Krotov, Denis S},
  journal={IEEE Transactions on Information Theory},
  volume={53},
  number={4},
  pages={1532--1537},
  year={2007},
  publisher={IEEE}
}

@article{Nechaev,
  title={Weighted modules and representations of codes},
  author={Honold, Th. and Nechaev, Aleksandr Aleksandrovich},
  journal={Probl. Inf. Transm.},
  volume={35},
  number={3},
  pages={205--223},
  year={1999},
  publisher={}
}

@article{Shankar,
  title={On {BCH} codes over arbitrary integer rings},
  author={Shankar, Priti},
  journal={IEEE Transactions on Information Theory},
  volume={25},
  number={4},
  pages={480--483},
  year={1979},
  publisher={IEEE}
}

@article{ShiKrotov2019,
  title={On  {$\mathbb{Z}_p\mathbb{Z}_{p^k}$}-additive codes and their duality},
  author={Shi, Minjia and Wu, Rongsheng and Krotov, Denis S},
  journal={IEEE Transactions on Information Theory},
  volume={65},
  number={6},
  pages={3841--3847},
  year={2019},
  publisher={IEEE}
}

@article{ShiTwoHomWeight,
  author    = {Minjia Shi and
               Thomas Honold and
               Patrick Sol{\'{e}} and
               Yunzhen Qiu and
               Rongsheng Wu and
               Zahra Sepasdar},
  title     = {The geometry of two-weight codes over {$\mathbb{Z}_{p^m}$}},
  journal   = {{IEEE} Transactions on Information Theory},
  volume    = {67},
  number    = {12},
  pages     = {7769--7781},
  year      = {2021}
}

@article{H07,
  author    = {Horadam, K. J.},
  title     = {{H}adamard matrices and their applications},
  journal   = {Princeton University Press},
  volume    = {},
  number    = {},
  pages     = {},
  year      = {2007}
}

@article{SBT98,
  author    = {Smith, E. D. J. and Blaikie, R. J. and Taylor D. P.},
  title     = {Performance enhancement of spectral-amplitude coding optical CDMA using pulse-position modulation},
  journal   = {{IEEE} Transactions on Communications},
  volume    = {46},
  number    = {},
  pages     = {1176--1185},
  year      = {1998}
}

@article{HYT04,
  author    = {Huang, J. F. and Yang, C. C. and Tseng, S. P.},
  title     = {Complementary {W}alsh-{H}adamard coded optical {CDMA} coder/decoders structured over arrayed-waveguide grating routers},
  journal   = {Opt. Commun.},
  volume    = {229},
  number    = {},
  pages     = {241--248},
  year      = {2004}
}

@article{Nyb91,
  author    = {Nyberg, K.},
  title     = {Perfect nonlinear {S}-boxes},
  journal   = {{EUROCRYPT-91}},
  volume    = {LNCS 547},
  number    = {},
  pages     = {378--385},
  year      = {1991},
 publisher={Springer}
}

@article{YLL03,
  author    = {Yu, G. J. and Lu, C. S. and Liao, H. Y.},
  title     = {A message-based cocktail watermarking system},
  journal   = {Pattern Recognition},
  volume    = {36},
  number    = {},
  pages     = {957--968},
  year      = {2003}
}

\end{document}